\documentclass{article}
\usepackage[preprint]{neurips_2026}
\usepackage{amsmath}
\usepackage{amsthm}
\usepackage[utf8]{inputenc} 
\usepackage[T1]{fontenc}    
\usepackage{hyperref}       
\usepackage{url}            
\usepackage{booktabs}       
\usepackage{amsfonts}       
\usepackage{nicefrac}       
\usepackage{microtype}      
\usepackage{xcolor}         
\usepackage{graphicx}
\newtheorem{remark}{Remark}

\newtheorem{theorem}{Theorem}

\title{Geometric Algebra Meets Cartesian Tensors: Higher-Order Equivariance for Interatomic Potentials}

\author{%
  \small 
  \begin{tabular}{c@{\hskip 0.6in}c}
    \textbf{\normalsize Can Polat} & \textbf{\normalsize Erchin Serpedin} \\
    \footnotesize Dept. of Electrical \& Computer Eng. & \footnotesize Dept. of Electrical \& Computer Eng. \\
    \footnotesize Texas A\&M Univ. & \footnotesize Texas A\&M Univ. \\
    \footnotesize College Station, TX, USA & \footnotesize College Station, TX, USA \\
    \texttt{\footnotesize can.polat@tamu.edu} & \texttt{\footnotesize eserpedin@tamu.edu} \\
    \\[0.05in] 
    \textbf{\normalsize Mustafa Kurban}\thanks{Corresponding author.} & \textbf{\normalsize Hasan Kurban}\footnotemark[1] \\
    \footnotesize Dept. of Prosthetics \& Orthotics & \footnotesize College of Science \& Eng. \\
    \footnotesize Ankara Univ., Ankara, Turkey & \footnotesize Hamad Bin Khalifa Univ. \\
    \footnotesize \emph{and} Texas A\&M Univ. at Qatar & \footnotesize Doha, Qatar \\
    \footnotesize Doha, Qatar & \texttt{\footnotesize hkurban@hbku.edu.qa} \\
    \texttt{\footnotesize kurbanm@ankara.edu.tr} &
  \end{tabular}
}

\begin{document}

\maketitle

\begin{abstract}
$\mathrm{Cl}(3,0)$ interatomic potentials, despite their algebraic elegance, predict force magnitudes accurately but force directions poorly: across ten rMD17 molecules, every $L \leq 1$ baseline in our twelve-model study attains aggregate force-cosine similarity below $0.25$. The cause is structural. The geometric product of two vectors in $\mathbb{R}^3$ realises only the $L{=}0$ and $L{=}1$ pieces of its irrep content, leaving the symmetric-traceless rank-2 component absent from the per-edge bilinear that drives every message-passing layer. We close the gap with \textbf{CliffordSTF}, which couples the Clifford multivector to closed-form symmetric-traceless tensor tracks at ranks two and three through bilinear cross-track contractions, using a single learned bilinear and no Clebsch--Gordan tables, Wigner-$D$ matrices, or e3nn calls. On rMD17, CliffordSTF raises aggregate force-cosine similarity from $0.055$ (base Clifford) to $0.551$, an order-of-magnitude relative directional gain at improved magnitude accuracy (force MAE $15.8\%$ lower; energy MAE $10.9\%$ lower), and outperforms every CG-free or body-ordered baseline we evaluate (all $\leq 0.17$). On catalysis, CliffordSTF attains the best out-of-distribution S2EF energy MAE on OC22 in our study and the best in-distribution energy MAE among $L \geq 2$ methods on OC22 IS2RE. An eleven-variant ablation confirms the two tracks are empirically complementary: neither alone approaches the hybrid.
\end{abstract}

\section{Introduction}\label{sec:intro}

\begin{figure}[!ht]
    \centering
    \includegraphics[width=1\linewidth]{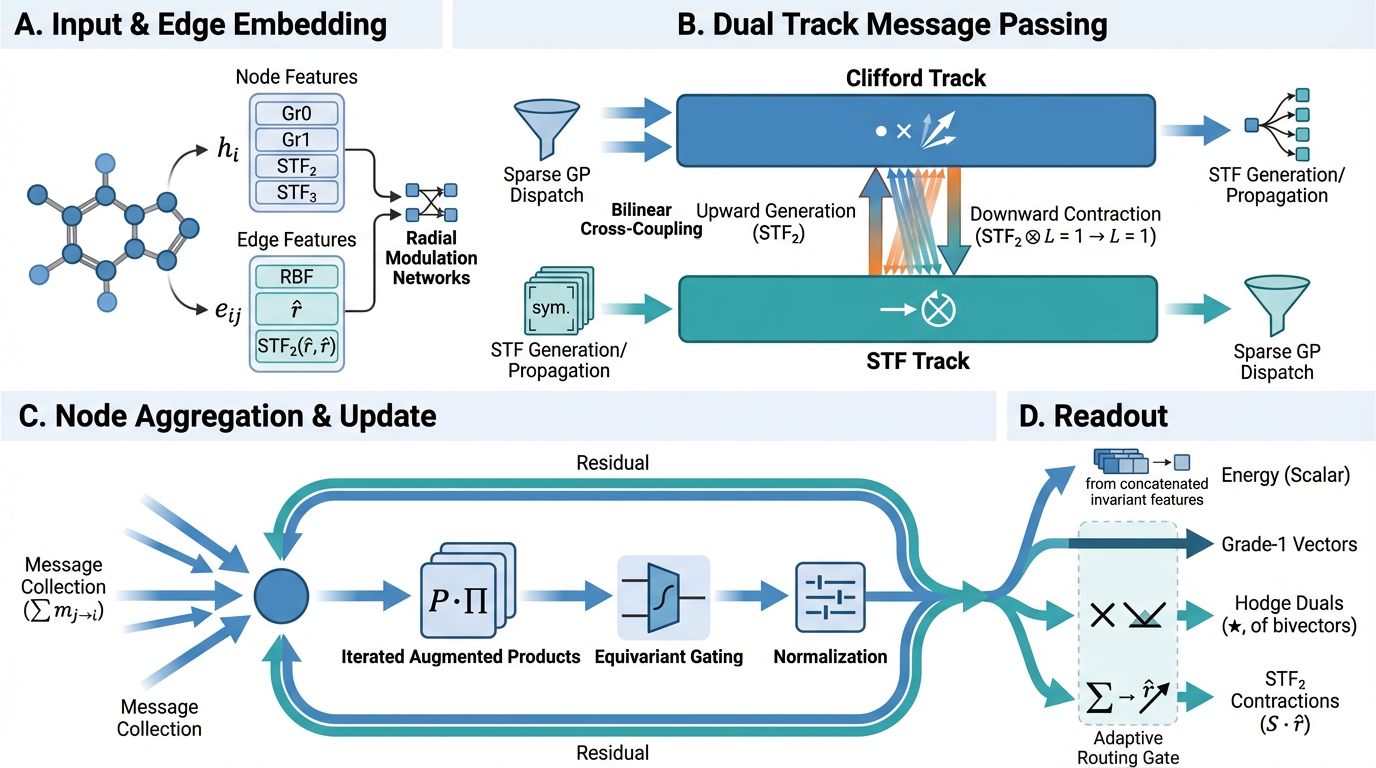}
    \caption{\textbf{Schematic of the CliffordSTF pipeline}. (A) Atoms and edges are embedded with multivector and STF features, including explicit $L{=}2$ edge descriptors. (B) The core message-passing block uses a dual-track design where the Clifford track (geometric products) and STF track are coupled via bilinear generation and contraction operations. (C) Node updates incorporate iterated augmented products for multi-body correlations followed by equivariant gating. (D) Energy and forces are predicted from a combination of invariant and vector sources, including $L{=}2$ contractions ($S \cdot \hat{r}$) that capture angular force components inaccessible to $L \leq 1$ features.}
    \label{fig:architecture}
\end{figure}

Machine learning interatomic potentials (MLIPs) have emerged as practical surrogates for first-principles electronic structure calculations, enabling molecular dynamics at scales inaccessible to density functional theory \citep{unke2021machine}. A central design challenge is encoding the symmetries of physical law: energies must be invariant and forces equivariant under rotations and translations. Force prediction, in particular, decomposes into a magnitude component and a directional component, and directional accuracy governs whether trajectories stay on the correct potential-energy basin even when scalar errors look small. The dominant approach to building such equivariant features constructs them from spherical harmonics composed through tensor products weighted by Clebsch--Gordan (CG) coefficients \citep{thomas2018tensor,batzner2022nequip,batatia2022mace}, but the cost of high-order couplings scales steeply with the maximum angular momentum order $L$, motivating both efficient implementations \citep{geiger2022e3nn,passaro2023escn} and CG-free reformulations based on irreducible Cartesian tensors \citep{simeon2023tensornet,zaverkin2024ictp,wang2024n}.

Clifford (geometric) algebra offers a third route. Architectures in this lineage embed vectors into the eight-dimensional algebra $\mathrm{Cl}(3,0)$ and replace learned tensor products with the single closed-form geometric product, which acts uniformly on scalars, vectors, bivectors, and pseudoscalars while preserving equivariance by construction \citep{ruhe2023cgenn,ruhe2023gcan,brehmer2023geometric,brandstetter2022clifford}. The substrate is appealing: it removes the CG bookkeeping, fuses grades into a single carrier, and admits a clean algebraic semantics. We argue, however, that it carries a structural expressivity gap that has not been named in prior work. The geometric product of two vectors in $\mathbb{R}^3$ produces a scalar and a bivector. In the language of $\mathrm{SO}(3)$ irreducible representations, these realise only the $L{=}0$ and $L{=}1$ pieces of the decomposition $\mathbf{1} \otimes \mathbf{1} = \mathbf{0} \oplus \mathbf{1} \oplus \mathbf{2}$; the five-dimensional symmetric-traceless $L{=}2$ piece sits outside the algebra entirely. Quadrupolar charge distributions, $d$-orbital interactions, and anisotropic stress all live naturally at $L{=}2$, and a Clifford backbone must therefore reconstruct them indirectly through repeated products and stacked layers, ceding ground to spherical-harmonic architectures whenever $L{\geq}2$ structure carries real signal.

\textbf{CliffordSTF} closes this gap by joining the two CG-free lineages reviewed above. The Cartesian-tensor side supplies closed-form symmetric-traceless tensor (STF) tracks at ranks two and three; the Clifford side supplies the multivector substrate; the two are coupled through bilinear cross-track contractions whose only learned component is the geometric product itself, implemented as an $8 \times 8$ Cayley table. Our contributions are:
\begin{itemize}
    \item \textbf{Diagnosis.} A controlled ablation isolates the missing $L \geq 2$ substrate, rather than capacity, depth, or readout choice, as the dominant driver of the directional failure mode shared across $L \leq 1$ baselines on rMD17.
    \item \textbf{Architecture.} We give a closed-form construction whose per-edge bilinear spans the $\mathrm{SO}(3)$-irrep content of spherical harmonics through $L{=}3$. The construction uses one learned bilinear together with classical symmetric-traceless contractions, with each induced map verified against Wigner-$3j$ symbols.
    \item \textbf{Empirics.} CliffordSTF outperforms every CG-free baseline we evaluate on rMD17 force-cosine similarity, attains the best in- and out-of-distribution energy MAE on OC22 S2EF in our study, and achieves the best in-distribution energy MAE among $L \geq 2$ methods on OC22 IS2RE. An eleven-variant ablation shows the two tracks are empirically complementary: neither alone approaches the hybrid.
\end{itemize}
\section{Related Work}\label{sec:related}

\paragraph{Equivariant GNNs for atomistic modelling.}
Invariant architectures such as SchNet~\citep{schutt2017schnet} and DimeNet++~\citep{gasteiger2020dimenetpp} encode atomic environments through radial and angular features but discard directional information at the representation level. Equivariant models restore this by constructing features that transform as irreducible representations of $\mathrm{SO}(3)$: NequIP~\citep{batzner2022nequip} propagates tensor field features through learned Clebsch--Gordan tensor products, MACE~\citep{batatia2022mace} achieves data efficiency via higher-order symmetric contractions, and PaiNN~\citep{schutt2021painn}, TorchMD-NET~\citep{tholke2022torchmdnet}, and ViSNet~\citep{wang2024visnet} operate on equivariant vector features ($L{=}1$) with favourable speed but limited angular resolution. EquiformerV2~\citep{liao2024equiformerv2} combines equivariant attention with eSCN convolutions~\citep{passaro2023escn} supporting tensor products up to $L{=}6$ at substantial computational cost. FAENet~\citep{duval2023faenet} achieves equivariance through stochastic frame averaging rather than architectural symmetry constraints, predicting forces directly from a learned head rather than via energy gradients. All CG-based methods share a fundamental bottleneck: tensor product cost grows combinatorially with $L$.

\paragraph{CG-free equivariant architectures.}
Several approaches achieve equivariant tensor interactions without Clebsch--Gordan coefficients. ICTP~\citep{zaverkin2024ictp} constructs products of irreducible Cartesian tensors (symmetric traceless tensors) within a MACE-like architecture. GotenNet~\citep{aykent2025gotennet} uses geometry-aware tensor attention via inner products on high-degree steerable features. Gaunt tensor products~\citep{luo2024gaunt} exploit spherical-harmonic addition theorems to reduce CG complexity from $\mathcal{O}(L^6)$ to $\mathcal{O}(L^3)$. TensorNet~\citep{simeon2023tensornet} replaces CG products with rank-2 Cartesian tensor operations, CACE~\citep{cheng2024cace} reformulates the atomic cluster expansion in Cartesian coordinates, and HotPP~\citep{wang2024n} passes reducible Cartesian tensor messages of arbitrary rank. The mathematical foundation underlying this entire line of work is the classical equivalence between symmetric traceless Cartesian tensors and $\mathrm{SO}(3)$ spherical irreps \citep{coope1965irreducible,applequist1989traceless,thorne1980multipole,stone2013theory}, recently placed on rigorous algorithmic footing via an orthonormal change-of-basis construction up to rank 9~\citep{shao2025orthonormal}. Our Theorem~\ref{thm:cg_completeness} operationalises this equivalence at the architectural level through the specific combination of the $\mathrm{Cl}(3,0)$ geometric product with closed-form STF operations. CliffordSTF shares the STF primitive with ICTP but differs in a substantive way: ICTP replaces CG products with a single Cartesian substrate, while CliffordSTF runs the Clifford multivector and STF tracks in parallel and couples them through a bidirectional cross-track bilinear. This two-substrate design is architecturally distinct, and our ablation shows the cross-track coupling is essential for end-to-end training stability, with its removal approximately doubling both the mean and standard deviation of energy MAE across seeds (Supplementary~\ref{supp:ictp}).

\paragraph{Clifford algebra in equivariant networks.}
Clifford algebra has been applied to equivariant networks across several domains: geometric Clifford algebra networks~\citep{ruhe2023gcan} introduced $\mathrm{Pin}(p,q,r)$ group layers, Clifford group equivariant neural networks~\citep{ruhe2023cgenn} extended equivariance to the full Clifford group, Clifford neural layers~\citep{brandstetter2022clifford} applied multivector feature maps to PDE surrogates, GATr~\citep{brehmer2023geometric} built a transformer over the projective algebra $\mathrm{Cl}(3,0,1)$, Clifford-steerable CNNs~\citep{zhdanov2024clifford} construct equivariant kernels on pseudo-Riemannian manifolds, and CG-EGNN~\citep{tran2024cgegnn} combined Clifford layers with high-order message passing. The structural gap identified in Section~\ref{sec:intro}, namely that the geometric product on $\mathbb{R}^3$ vectors leaves $L \geq 2$ inaccessible at the per-edge bilinear regardless of depth, is not, to our knowledge, named in any of these works. Two clarifications guard against likely misreadings. First, higher algebraic \emph{grade} in $\mathrm{Cl}(3,0)$ does not correspond to higher angular momentum $L$: grade-2 bivectors transform as $L{=}1$ pseudovectors under $\mathrm{SO}(3)$, not $L{=}2$. Second, the same deficiency carries over to the projective algebra $\mathrm{Cl}(3,0,1)$ used in GATr, where the additional null basis encodes translational rather than angular content and the $\mathrm{SO}(3)$ decomposition remains $2 \cdot L{=}0 \oplus 2 \cdot L{=}1$ (Supplementary~\ref{supp:gatr}).

\paragraph{Force directional fidelity as a metric.}
Standard magnitude metrics can obscure the representational differences above: low force MAE does not reliably predict simulation quality~\citep{fu2023forces}, and hyperparameters optimising force accuracy can degrade energy accuracy~\citep{kovacs2023mace}. Force cosine similarity, part of the OC20 benchmark since its introduction but not standardly reported on rMD17, exposes the gap directly. We bring it to small-molecule evaluation and observe a representational stratification consistent with our theoretical framing: $L \leq 1$ direct-force architectures cluster at poor directional alignment (all below $0.15$ at matched budget), CG-free $L \geq 2$ architectures (MACE, ICTP) improve only marginally without their recommended hyperparameters, and only fully spherical-harmonic $L \geq 2$ architectures (EquiformerV2) realise the high-fidelity regime. CliffordSTF closes most of this gap without spherical harmonics.
\section{CliffordSTF}\label{sec:method}

\subsection{The \texorpdfstring{$L \leq 1$}{L <= 1} Limitation at the Bilinear Level}\label{sec:limitation}

The diagnosis of Section~\ref{sec:intro} can be made precise at the algebraic level. Under $\mathrm{SO}(3)$ the graded components of $\mathrm{Cl}(3,0)$ decompose as $2\cdot L{=}0 \oplus 2\cdot L{=}1$: grade~0 and grade~3 carry $L{=}0$, while grade~1 and the Hodge dual of grade~2 carry $L{=}1$. No grade contains the 5-dimensional $L{=}2$ irrep \citep{doran2003geometric,lounesto2001clifford} (see Supplementary~\ref{supp:prelim}). Because the geometric product is closed within the algebra, $L \geq 2$ content can arise only nonlinearly through stacked products; Section~\ref{sec:stf_channels} supplies it directly through closed-form maps that span $\mathbf{1}\otimes\mathbf{1}$ at the per-edge bilinear.

\begin{remark}[Equivariance scope]\label{rem:so3_vs_o3}
All equivariance claims are $\mathrm{SO}(3)$. Under $\mathrm{SO}(3)$, polar vectors (grade~1) and axial vectors (Hodge of grade~2) transform identically. Under $\mathrm{O}(3)$ they carry parity $-1$ and $+1$ respectively; full $\mathrm{O}(3)$ equivariance would require explicit parity tracking (Supplementary~\ref{supp:so3_o3}).
\end{remark}

\subsection{Symmetric Traceless Tensor Channels}\label{sec:stf_channels}

A rank-$\ell$ symmetric traceless tensor in $\mathbb{R}^3$ has $2\ell+1$ independent components and transforms as the $L{=}\ell$ irrep of $\mathrm{SO}(3)$ \citep{coope1965irreducible,applequist1989traceless,stone2013theory,shao2025orthonormal}. We use two:
\[
\mathrm{STF}_2(u,v)_{ij} = \tfrac{1}{2}(u_iv_j + u_jv_i) - \tfrac{1}{3}(u \cdot v)\,\delta_{ij}, \qquad
\mathrm{STF}_3(S,v)_{ijk} = T^{\text{sym}}_{ijk} - \tfrac{1}{5}\bigl(\delta_{ij}A_k + \delta_{ik}A_j + \delta_{jk}A_i\bigr),
\]
where $S$ is symmetric traceless, $T^{\text{sym}}_{ijk} = \tfrac{1}{3}(S_{ij}v_k + S_{ik}v_j + S_{jk}v_i)$, and $A_k = \tfrac{2}{3}S_{kl}v_l$. Coefficient derivations and DOF checks are in Supplementary~\ref{supp:stf_coefficients}. STF$_2$ is stored as the 5 components $[S_{xx}, S_{xy}, S_{xz}, S_{yy}, S_{yz}]$ with $S_{zz}=-S_{xx}-S_{yy}$ implicit; STF$_3$ is stored as 7 independent components.

\textbf{Constructive Clebsch--Gordan equivalence.} The equivalence between symmetric traceless Cartesian tensors and spherical-harmonic irreps of $\mathrm{SO}(3)$ is classical and has been placed on rigorous algorithmic footing via orthonormal change-of-basis matrices up to rank 9. Our contribution is architectural: we identify the specific combination of the $\mathrm{Cl}(3,0)$ geometric product with closed-form STF operations as sufficient to span the full Clebsch--Gordan decomposition for irrep tensor products through $L{=}3$ at the per-edge bilinear.

\begin{theorem}[Constructive CG recovery via geometric product and STF operations]\label{thm:cg_completeness}
Let $u, v \in \mathbb{R}^3$ carry the $L{=}1$ irrep of $\mathrm{SO}(3)$. The Clebsch--Gordan decomposition $\mathbf{1} \otimes \mathbf{1} = \mathbf{0} \oplus \mathbf{1} \oplus \mathbf{2}$ is recovered in closed form by
\[
[u \otimes v]_{L=0} = \langle uv \rangle_0, \qquad
[u \otimes v]_{L=1} = \star\langle uv \rangle_2, \qquad
[u \otimes v]_{L=2} = \mathrm{STF}_2(u,v),
\]
where $\langle \cdot \rangle_k$ extracts grade~$k$ of the $\mathrm{Cl}(3,0)$ geometric product $uv$, $\star$ is the Hodge dual, and $\mathrm{STF}_2$ is the classical symmetric traceless projector. The three components are in bijective correspondence with the irrep components of $u_iv_j$ under the canonical embedding $u_iv_j = \tfrac{1}{3}\delta_{ij}(u\cdot v) + \tfrac{1}{2}\varepsilon_{ijk}(\star\langle uv \rangle_2)_k + \mathrm{STF}_2(u,v)_{ij}$, and each map is $\mathrm{SO}(3)$-equivariant.
\end{theorem}

\begin{proof}[Proof sketch]
Expand $u_iv_j = \tfrac{1}{3}\delta_{ij}(u\cdot v) + \tfrac{1}{2}(u_iv_j - u_jv_i) + \mathrm{STF}_2(u,v)_{ij}$. The antisymmetric part equals $\tfrac{1}{2}\varepsilon_{ijk}(u\times v)_k$, and $u\times v = \star(u\wedge v)$ under the Hodge isomorphism $\Lambda^2\mathbb{R}^3 \cong \mathbb{R}^3$. Equivariance follows from $\mathrm{SO}(3)$-equivariance of the geometric product, the Hodge star, and symmetric/trace-removal operations. Full coordinates in Supplementary~\ref{supp:cg_proof}.
\end{proof}

Higher-order couplings ($\mathbf{2}\otimes\mathbf{1}$, $\mathbf{2}\otimes\mathbf{2}$, $\mathbf{3}\otimes\mathbf{1}$, with the $L{=}4$ component of $\mathbf{2}\otimes\mathbf{2}$ outside our $L_{\max}{=}3$ cut) are recovered by analogous STF contractions, enumerated with explicit coordinates in Supplementary~\ref{supp:cg_proof}. The cut is set at $L_{\max}{=}3$ to balance new $L{\geq}2$ capacity against the matched $10^6$-parameter budget; an STF$_4$ extension is straightforward but unnecessary for closing the directional gap diagnosed in \S\ref{sec:limitation}. Each constructive map is a polynomial in its inputs with constant rational coefficients matching the classical detracing rule $1/(2\ell-1)$ in $d=3$. We verify each map against a direct Wigner $3j$-symbol reference to machine precision (Supplementary~\ref{supp:cg_test}), so the architectural realisation reproduces the orthonormal equivalence in the specific cases used by CliffordSTF.

\subsection{Architecture}\label{sec:architecture}

Figure~\ref{fig:architecture} summarises the pipeline. CliffordSTF carries two parallel feature substrates per atom: a Clifford multivector ($2 \cdot L{=}0 \oplus 2 \cdot L{=}1$, dimension 8) and optional symmetric-traceless tensor channels at ranks 2 and 3. We evaluate three configurations: \textbf{Base} (Clifford only, 8D), \textbf{+STF$_2$} (13D), and \textbf{+STF$_3$} (20D). With STF channels, Hodge, routing, and cross-track coupling all disabled, the Base configuration recovers a vanilla Clifford model bit-for-bit, ensuring controlled ablation. A per-layer grade schedule activates grades progressively, so early layers see grade-sparse geometric products and deeper layers see the full multivector. The only nonstandard primitives are the $8 \times 8$ Cayley table and the closed-form STF operations of \S\ref{sec:stf_channels}; everything else uses standard PyTorch.

\textbf{Embedding.} Atomic numbers map to scalar features via a learned embedding; STF channels are initialised with small Gaussian noise ($\sigma{=}10^{-3}$). Edges carry grade-0 RBF features, grade-1 direction features, and STF$_2$ features from $\mathrm{STF}_2(\hat{r}_{ij}, \hat{r}_{ij})$, all radially modulated, assembled into an edge multivector $\tilde{e}_{ij}$ in $\mathrm{Cl}(3,0)$. The tilde distinguishes this edge feature from algebra basis bivectors (Supplementary~\ref{supp:edge_embed}).

\textbf{Dual-track message passing.} Let $\mathrm{GP}$ denote the $\mathrm{Cl}(3,0)$ geometric product. The Clifford-track message is
\[
m^{\mathrm{Cl}}_{j \to i} = \alpha_{ij} \bigl[W_1 \cdot \mathrm{GP}(h_j^{\mathrm{Cl}}, \tilde{e}_{ij}) + W_2 \cdot \mathrm{GP}(\tilde{e}_{ij}, h_j^{\mathrm{Cl}}) + W_{\mathrm{skip}} \cdot h_j^{(0)}\bigr],
\]
with grade-sparse GP dispatch and an invariant scalar attention weight $\alpha_{ij}$ from multi-head dot-product attention on grade-0 features with an additive RBF bias. STF$_2$ messages combine generation from grade-1 features ($\mathrm{STF}_2(h_j^{\mathrm{vec}}, \hat{r}_{ij})$) with radially-modulated propagation of the sender's STF$_2$ features; STF$_3$ messages are analogous, generated via $\mathrm{STF}_3(h_j^{\mathrm{STF}_2}, \hat{r}_{ij})$.

\textbf{Cross-track coupling.} The sender's STF$_2$ features contract with its grade-1 vectors via $S \cdot v$, injecting $L{=}2$ information into the Clifford track:
\[
m^{\mathrm{Cl}}_{j \to i}[\text{grade-1}] \mathrel{+}= W_{\mathrm{cross}} \cdot \mathrm{contract}(h_j^{\mathrm{STF}_2}, h_j^{\mathrm{vec}}).
\]
This creates bidirectional flow: grade-1 vectors generate STF$_2$ (upward) while STF$_2$ feeds back into grade-1 through contraction (downward). At the many-body stage an augmented product couples all tracks through six bilinear terms, each corresponding bijectively to a CG channel in Theorem~\ref{thm:cg_completeness} and its higher-rank extensions (enumerated in Supplementary~\ref{supp:aug_prod}). Without cross-track coupling, STF parameters receive only weak invariant-norm gradients and degrade training; this is confirmed by the ablation (\S\ref{sec:results_scalar}).

\textbf{Update.} Iterated augmented products provide multi-body correlations up to body order 4. The combined output is then projected by a block-diagonal CliffordSTFLinear, equivariantly gated (SiLU on scalars, norm-gating on $L \geq 1$ \citep{weiler20183d}), passed through a residual connection, and normalised track-wise (Supplementary~\ref{supp:normalization}).

\textbf{Adaptive routing (optional).} A routing module applies per-atom multiplicative gates to the STF tracks between layers, reducing effective angular bandwidth for atoms in symmetric environments; the Clifford track is always active. Static (per-atom-type) and learned (MLP on invariant descriptors) variants are supported (Supplementary~\ref{supp:routing}).

\subsection{Readout}\label{sec:readout}

\textbf{Energy.} Per-atom energies are computed from invariant features: grade-0 scalars, GP self-interaction scalars (grade-0 component of $\mathrm{GP}(W h^{\mathrm{Cl}}, h^{\mathrm{Cl}})$), and per-channel Frobenius norms $\lVert h^{\mathrm{STF}_2}\rVert_F$ and $\lVert h^{\mathrm{STF}_3}\rVert_F$, concatenated and passed through an MLP. Per-layer readouts summed across interaction layers give a MACE-style multi-scale decomposition.

\textbf{Forces.} The per-atom head applies a linear projection over $2C$ vector channels obtained by concatenating grade-1 vectors and the Hodge duals of grade-2 bivectors (valid under $\mathrm{SO}(3)$; Remark~\ref{rem:so3_vs_o3}). This effectively doubles the force basis at no additional equivariant cost: the Hodge duals are a rearrangement of bivector features already present in the multivector. A per-edge head adds a scatter-summed STF$_2$ contribution $S_{ij}\hat{r}_{ij}$, magnitude-gated by an MLP on $\lVert h^{\mathrm{STF}_2}\rVert_F$. This explicit $L{=}2 \otimes L{=}1 \to L{=}1$ path captures angular force components inaccessible to $L \leq 1$ architectures.
\section{Experiments}\label{sec:experiments}

\subsection{Benchmarks and protocol}\label{sec:benchmarks}

We evaluate on five benchmarks across two tiers: three primary interatomic-potential benchmarks (\textbf{rMD17}, \textbf{OC20}, \textbf{OC22}) and two secondary scalar-property benchmarks (\textbf{QM9}, \textbf{Molecule3D}). We compare against twelve architecture families spanning invariant, vector-equivariant, tensor-field, body-ordered, frame-averaging, and Clifford-algebraic models. The Clifford-algebraic baseline is our base $\mathrm{Cl}(3,0)$, recovered from CliffordSTF by disabling all extensions, providing a controlled within-paper comparison. MACE, ICTP, and EquiformerV2 are additionally swept across $L_{\max}$ to trace angular resolution within a family. The full baseline list and per-family configurations are in Supplementary~\ref{supp:model_settings}. All models use a unified training framework at approximately $10^6$ parameters, Adam at learning rate $10^{-4}$, no weight decay, early stopping at patience 30, and no EMA or warm-up; per-benchmark training settings are in Supplementary~\ref{supp:protocol}. The protocol targets representational comparability at fixed capacity rather than per-architecture state-of-the-art, so several baselines do not reach their published accuracies under this fixed-budget regime; these cases are flagged in the tables and discussed in Supplementary~\ref{supp:results_caveats}. Metrics are energy MAE, force MAE, and force cosine similarity (units per benchmark in each table caption); all numbers are mean $\pm$ standard deviation across seeds.

\subsection{Force directional accuracy on rMD17}\label{sec:results_md17}

Force cosine similarity isolates directional fidelity from magnitude accuracy and is standard on OC20 but has not been previously reported on rMD17; we include it as our primary mechanistic diagnostic. Four models receive caveat marks in Table~\ref{tab:md17_fcos}: SchNet and DimeNet++ are strictly invariant (forces from energy gradients, marked gray), FAENet uses stochastic frame averaging rather than architectural equivariance ($\ast$), GotenNet is $L\!\geq\!2$ but did not converge under our fixed-budget protocol ($\ast$), and EquiformerV2 predicts absolute molecular energies so its rMD17 energy MAE is not comparable to the other rows ($\dagger$). Per-model justifications are in Supplementary~\ref{supp:results_caveats}; all four are excluded from the per-architecture directional narrative below.

Table~\ref{tab:md17_fcos} shows a clean separation by maximum angular momentum. EquiformerV2 at $L_{\max}\!\geq\!2$ attains aggregate force cosine similarity near $0.96$, CliffordSTF attains $0.551$, and every other direct-force model sits below $0.25$. Within the $L\!\leq\!1$ block the picture splits further: the vector-equivariant baselines (PaiNN, TorchMD-NET, ViSNet, NequIP) cluster near zero, while the body-ordered $L{=}1$ baselines (MACE, ICTP) reach at most $0.23$, still far below the CliffordSTF result.

The decisive within-paper comparison is Clifford versus CliffordSTF, which share message-passing infrastructure, training protocol, and parameter budget and differ only by the addition of STF tracks with cross-track coupling. \textbf{Force cosine similarity climbs from $0.055$ to $0.551$, an order-of-magnitude directional improvement.} The same comparison improves force MAE from $20.9$ to $17.6$\,meV/\AA{} ($15.8\%$ lower) and energy MAE from $34.1$ to $30.4$\,meV ($10.9\%$ lower): the STF extension improves directional fidelity and magnitude accuracy together, not at one another's expense. Per-molecule force and energy MAE appear in Supplementary~\ref{supp:md17_permol}.

\begin{table}[t]
\centering
\caption{\textbf{rMD17 force cosine similarity} per molecule and aggregate (mean $\pm$ std across 3 seeds; higher is better, $[-1,1]$ range). Top block: $L\!\geq\!2$ direct-force models. Middle block: $L\!\leq\!1$ direct-force models. Bottom block (gray): strictly invariant models, forces from energy gradients. $\dagger$ absolute-energy prediction (not energy-referenced---magnitude reflects binding energy, not error). GotenNet, FAENet, and MACE/ICTP force MAEs in the $10^3$\,meV/\AA{} range reflect non-convergence at the $10^6$-parameter budget. Per-model reasoning in Supplementary~\ref{supp:results_caveats}; per-molecule force/energy MAE in Supplementary~\ref{supp:md17_permol}.}
\label{tab:md17_fcos}
\setlength{\tabcolsep}{3pt}
\resizebox{\textwidth}{!}{%
\begin{tabular}{lccccccccccc}
\toprule
Model & Asp & Azo & Bnz & Eth & Mal & Nap & Par & Sal & Tol & Ura & Agg \\
\midrule
EquiformerV2 $L{=}4$$^{\dagger}$ & 0.920 & 0.958 & 0.986 & 0.987 & 0.976 & 0.980 & 0.922 & 0.949 & 0.976 & 0.956 & \textbf{0.961} $\pm$ 0.002 \\
EquiformerV2 $L{=}3$$^{\dagger}$ & 0.916 & 0.962 & 0.984 & 0.981 & 0.959 & 0.982 & 0.934 & 0.950 & 0.975 & 0.960 & 0.960 $\pm$ 0.002 \\
EquiformerV2 $L{=}2$$^{\dagger}$ & 0.911 & 0.956 & 0.986 & 0.961 & 0.953 & 0.983 & 0.933 & 0.946 & 0.971 & 0.958 & 0.956 $\pm$ 0.002 \\
CliffordSTF (ours)   & 0.179 & 0.721 & 0.862 & 0.586 & 0.594 & 0.794 & 0.319 & 0.406 & 0.685 & 0.362 & \textbf{0.551} $\pm$ 0.064 \\
MACE $L{=}3$         & 0.161 & 0.303 & 0.210 & 0.036 & 0.202 & 0.186 & 0.156 & 0.181 & 0.082 & 0.261 & 0.178 $\pm$ 0.017 \\
MACE $L{=}2$         & 0.083 & 0.239 & 0.252 & -0.018 & 0.128 & 0.263 & 0.171 & 0.135 & 0.091 & 0.193 & 0.154 $\pm$ 0.017 \\
ICTP $L{=}3$         & 0.175 & 0.344 & 0.198 & 0.021 & 0.199 & 0.249 & 0.180 & 0.096 & 0.064 & 0.218 & 0.174 $\pm$ 0.030 \\
ICTP $L{=}2$         & 0.176 & 0.232 & 0.290 & -0.101 & 0.107 & 0.266 & 0.196 & 0.202 & 0.149 & 0.134 & 0.165 $\pm$ 0.051 \\
GotenNet    & 0.088 & -0.021 & -0.022 & 0.089 & 0.062 & 0.003 & 0.026 & 0.072 & -0.090 & 0.093 & 0.030 $\pm$ 0.012 \\
\midrule
Clifford (ours) & 0.012 & 0.030 & 0.075 & -0.004 & 0.012 & 0.120 & 0.010 & 0.131 & 0.047 & 0.117 & 0.055 $\pm$ 0.004 \\
MACE $L{=}1$         & 0.181 & 0.370 & 0.312 & 0.101 & 0.247 & 0.284 & 0.184 & 0.262 & 0.138 & 0.196 & 0.228 $\pm$ 0.022 \\
ICTP $L{=}1$         & 0.180 & 0.242 & 0.283 & 0.026 & 0.210 & 0.177 & 0.210 & 0.211 & 0.073 & 0.238 & 0.185 $\pm$ 0.012 \\
PaiNN                & -0.006 & 0.088 & 0.003 & 0.085 & -0.022 & 0.032 & -0.002 & 0.032 & -0.052 & -0.022 & 0.014 $\pm$ 0.032 \\
TorchMD-NET          & 0.011 & 0.029 & 0.048 & 0.204 & -0.069 & -0.052 & 0.001 & 0.064 & 0.034 & 0.029 & 0.030 $\pm$ 0.006 \\
ViSNet               & -0.023 & 0.045 & -0.126 & 0.049 & 0.052 & 0.008 & 0.018 & -0.006 & 0.079 & -0.044 & 0.005 $\pm$ 0.025 \\
NequIP    & 0.064 & 0.155 & 0.070 & -0.019 & 0.063 & 0.065 & 0.019 & -0.140 & -0.064 & -0.107 & 0.011 $\pm$ 0.001 \\
FAENet      & 0.024 & 0.039 & 0.518 & 0.088 & 0.154 & 0.132 & 0.096 & 0.041 & 0.120 & 0.211 & 0.142 $\pm$ 0.058 \\
\midrule
\color{gray}{DimeNet++}   & \color{gray}{0.029} & \color{gray}{0.053} & \color{gray}{0.050} & \color{gray}{0.032} & \color{gray}{0.020} & \color{gray}{0.030} & \color{gray}{0.054} & \color{gray}{0.055} & \color{gray}{0.047} & \color{gray}{0.003} & \color{gray}{0.037 $\pm$ 0.013} \\
\color{gray}{SchNet}      & \color{gray}{-0.029} & \color{gray}{-0.098} & \color{gray}{-0.013} & \color{gray}{0.159} & \color{gray}{0.068} & \color{gray}{-0.112} & \color{gray}{-0.026} & \color{gray}{-0.065} & \color{gray}{-0.086} & \color{gray}{-0.038} & \color{gray}{-0.024 $\pm$ 0.007} \\
\bottomrule
\end{tabular}%
}
\end{table}

\subsection{Catalysis: OC22 and OC20}\label{sec:results_catalysis}

\textbf{OC22.} CliffordSTF is competitive across the OC22 benchmark. On IS2RE, it attains the best in-distribution MAE among $L\!\geq\!2$ architectures ($4.09$\,eV; next-best $L\!\geq\!2$ GotenNet at $4.32$\,eV) and matches EquiformerV2 on the OOD split within seed variance, behind GotenNet ($5.77$\,eV, the strongest IS2RE-OOD model in the study); CliffordSTF additionally attains the best $L\!\geq\!2$ EwT out-of-distribution. On S2EF, CliffordSTF attains the best out-of-distribution energy MAE in the study and is within seed variance of the best in-distribution energy MAE (held by base Clifford at $2.844$ vs $2.866$\,eV); EquiformerV2 leads on force cosine similarity, with force MAEs clustering tightly across the top three models. Full numbers in Table~\ref{tab:oc22}.

\begin{table}[t]
\centering
\caption{\textbf{OC22 main results} (mean $\pm$ std across 3 seeds). Best value per column in \textbf{bold}. EquiformerV2 is our controlled-protocol result at $L_{\max}{=}4$.}
\label{tab:oc22}
\setlength{\tabcolsep}{3pt}
\resizebox{\textwidth}{!}{%
\begin{tabular}{lcccccccc}
\toprule
& \multicolumn{2}{c}{IS2RE MAE (eV)} & \multicolumn{3}{c}{S2EF $\mathrm{val}_{\text{id}}$} & \multicolumn{3}{c}{S2EF $\mathrm{val}_{\text{ood}}$} \\
\cmidrule(lr){2-3}\cmidrule(lr){4-6}\cmidrule(lr){7-9}
Model & $\mathrm{val}_{\text{id}}$ & $\mathrm{val}_{\text{ood}}$ & $e_\mathrm{MAE}$ (eV) & $f_\mathrm{MAE}$ (eV/\AA) & $f_\mathrm{cos}\uparrow$ & $e_\mathrm{MAE}$ (eV) & $f_\mathrm{MAE}$ (eV/\AA) & $f_\mathrm{cos}\uparrow$ \\
\midrule
CliffordSTF (ours) & 4.089 $\pm$ 0.261 & 6.463 $\pm$ 0.310 & 2.866 $\pm$ 0.093 & 0.073 $\pm$ 0.000 & 0.050 $\pm$ 0.004 & \textbf{5.083} $\pm$ 0.166 & \textbf{0.070} $\pm$ 0.000 & 0.050 $\pm$ 0.003 \\
Clifford (ours)    & 4.422 $\pm$ 0.862 & 6.726 $\pm$ 0.656 & \textbf{2.844} $\pm$ 0.100 & \textbf{0.072} $\pm$ 0.000 & 0.059 $\pm$ 0.002 & 5.375 $\pm$ 0.157 & 0.071 $\pm$ 0.001 & 0.058 $\pm$ 0.002 \\
EquiformerV2       & 4.905 $\pm$ 0.566 & 6.619 $\pm$ 0.445 & 3.369 $\pm$ 0.066 & \textbf{0.072} $\pm$ 0.000 & \textbf{0.067} $\pm$ 0.001 & 5.262 $\pm$ 0.076 & \textbf{0.070} $\pm$ 0.000 & \textbf{0.062} $\pm$ 0.002 \\
PaiNN              & \textbf{3.516} $\pm$ 0.292 & 6.712 $\pm$ 0.140 & 5.306 $\pm$ 0.030 & 0.076 $\pm$ 0.000 & 0.048 $\pm$ 0.001 & 7.452 $\pm$ 0.254 & 0.078 $\pm$ 0.000 & 0.036 $\pm$ 0.001 \\
GotenNet           & 4.318 $\pm$ 0.300 & \textbf{5.772} $\pm$ 0.532 & 4.666 $\pm$ 0.310 & 0.073 $\pm$ 0.000 & 0.052 $\pm$ 0.001 & 7.200 $\pm$ 0.121 & 0.072 $\pm$ 0.001 & 0.044 $\pm$ 0.001 \\
TorchMD-NET        & 5.397 $\pm$ 0.000 & 8.312 $\pm$ 0.000 & 4.278 $\pm$ 0.000 & 0.075 $\pm$ 0.000 & 0.051 $\pm$ 0.000 & 6.232 $\pm$ 0.000 & 0.074 $\pm$ 0.000 & 0.043 $\pm$ 0.000 \\
DimeNet++          & 6.320 $\pm$ 4.141 & 9.897 $\pm$ 3.002 & 5.288 $\pm$ 0.030 & 0.076 $\pm$ 0.000 & 0.054 $\pm$ 0.002 & 8.717 $\pm$ 0.075 & 0.077 $\pm$ 0.001 & 0.038 $\pm$ 0.002 \\
SchNet             & 6.402 $\pm$ 1.688 & 9.424 $\pm$ 0.938 & 7.690 $\pm$ 1.508 & 0.077 $\pm$ 0.002 & 0.032 $\pm$ 0.007 & 12.016 $\pm$ 1.938 & 0.078 $\pm$ 0.002 & 0.024 $\pm$ 0.003 \\
MACE $L{=}2$       & 8.416 $\pm$ 0.702 & 22.569 $\pm$ 7.674 & 10.017 $\pm$ 0.395 & 0.079 $\pm$ 0.000 & 0.022 $\pm$ 0.004 & 15.415 $\pm$ 0.808 & 0.076 $\pm$ 0.000 & 0.022 $\pm$ 0.005 \\
ICTP $L{=}2$       & 8.641 $\pm$ 1.744 & 17.756 $\pm$ 4.145 & 12.447 $\pm$ 0.354 & 0.081 $\pm$ 0.001 & 0.015 $\pm$ 0.002 & 19.218 $\pm$ 0.232 & 0.079 $\pm$ 0.001 & 0.017 $\pm$ 0.002 \\
ICTP $L{=}1$       & 10.963 $\pm$ 1.612 & 20.810 $\pm$ 0.868 & 11.483 $\pm$ 0.881 & 0.080 $\pm$ 0.001 & 0.019 $\pm$ 0.001 & 17.263 $\pm$ 0.655 & 0.186 $\pm$ 0.154 & 0.019 $\pm$ 0.000 \\
MACE $L{=}1$       & 12.339 $\pm$ 4.540 & 15.800 $\pm$ 2.430 & 10.011 $\pm$ 0.018 & 0.078 $\pm$ 0.000 & 0.024 $\pm$ 0.000 & 15.737 $\pm$ 0.438 & 0.076 $\pm$ 0.000 & 0.023 $\pm$ 0.001 \\
NequIP             & 12.575 $\pm$ 4.889 & 14.865 $\pm$ 3.676 & 9.868 $\pm$ 0.148 & 0.079 $\pm$ 0.000 & 0.028 $\pm$ 0.000 & 15.809 $\pm$ 0.071 & 0.076 $\pm$ 0.000 & 0.027 $\pm$ 0.001 \\
\bottomrule
\end{tabular}%
}
\end{table}

\textbf{OC20.} On OC20, which uses a different distribution of catalysis systems than OC22, CliffordSTF maintains competitive results: in-distribution IS2RE MAE $0.703$\,eV (between DimeNet++ at $0.699$ and PaiNN at $0.720$), with the model in the top half of the field on every split. On S2EF, EquiformerV2 leads on force cosine similarity ($0.170$ vs $0.138$ for CliffordSTF), while the base Clifford model attains the best energy MAE ($1.89$ vs $2.11$ for CliffordSTF). Full per-split tables in Supplementary~\ref{supp:oc_full}.

\textbf{Scaling check.} Scaling CliffordSTF to $10^7$ parameters preserves stable training and yields the expected improvement on the catalysis task. On OC22 IS2RE, the best CliffordSTF seed reaches in-distribution MAE of $0.0162$\,eV/atom at $44$ epochs, a roughly threefold gain over its $10^6$-parameter result. EquiformerV2-L2 scales more steeply on the same benchmark ($0.0094$\,eV/atom, about sixfold), overtaking CliffordSTF at this larger capacity, while MACE-L2 stays essentially flat at $0.0779$\,eV/atom, consistent with its under-convergence at the smaller budget. On OC22 S2EF force MAE, CliffordSTF reaches $0.0769$\,eV/\AA{}, within $1.4\%$ of EquiformerV2-L2 at $0.0758$\,eV/\AA{} and ahead of MACE-L2 at $0.0831$\,eV/\AA{}. Base Clifford also trains stably to $10^7$ parameters and reaches $0.0773$\,eV/\AA{} force MAE on OC22 S2EF, essentially indistinguishable from CliffordSTF. The check confirms that the family extends to higher capacity without breaking, with no degradation observed on either the Clifford or CliffordSTF model. Details in Supplementary~\ref{supp:scaling_check}.

\subsection{Scalar properties, ablation, and cost}\label{sec:results_scalar}

On QM9 six-target aggregate MAE (mixed units across targets; the aggregate averages normalised targets), CliffordSTF ranks fifth of ten at $2.15$\,units, confirming that the STF extension does not degrade scalar prediction below vector-equivariant baselines. On Molecule3D HOMO--LUMO gap, CliffordSTF attains MAE $1.23$\,eV, behind FAENet and the strong mid-pack and ahead of DimeNet++, SchNet, and NequIP. FAENet's large advantage on this benchmark reflects its scalar-specialised design; we treat Molecule3D as out-of-scope for the architectural thesis and discuss in \S\ref{sec:limitations}. Full per-target tables are in Supplementary~\ref{supp:qm9_full}.

\textbf{Ablation.} An eleven-variant ablation on rMD17 (Table~\ref{tab:ablation}) isolates each track's contribution and shows the two lineages are empirically complementary. A Clifford-only configuration (vanilla $L{=}1$ Clifford, no STF anywhere) reaches $f_{cos}{=}0.030$; an STF-only configuration (STF tracks present at the output only, no STF inside the message-passing scaffold) reaches $f_{cos}{=}0.097$; only the hybrid CliffordSTF scaffold attains $f_{cos}{\in}[0.421, 0.506]$, an order of magnitude above Clifford-only and roughly fivefold above the readout-only tier. The fivefold separation between the readout-only and scaffold tiers isolates the dominant effect to $L\!\geq\!2$ representations participating throughout message passing, ruling out an explanation in which an $L\!\geq\!2$ readout alone is sufficient. Within the scaffold, individual flag toggles produce smaller effects that are often within seed-to-seed variance on headline $f_{cos}$, but they matter noticeably for energy MAE stability (e.g., Full no cross reaches $e_{\text{MAE}}{=}70.4$\,meV versus $46.0$ for Full with cross) and for the per-seed standard deviation. The cross-track coupling contribution is therefore better read as a stability lever than a raw-accuracy driver: the architectural claim is that the $L\!\geq\!2$ STF scaffold unlocks directional capacity, and cross-track coupling stabilises it for reliable end-to-end training.

\begin{table}
\centering
\caption{\textbf{CliffordSTF ablation on rMD17} (mean $\pm$ std across 6 seeds per (variant, molecule); molecules: aspirin, benzene, ethanol, salicylic acid). Twelve variants in three tiers: vanilla $L{=}1$ Clifford (bottom), readout-only STF (middle), and CliffordSTF-scaffold variants (top). Channel counts tuned per row to hold parameter counts at approximately $10^6$. Interpretation in \S\ref{sec:results_scalar}; flag definitions in Supplementary~\ref{supp:variants}.}
\label{tab:ablation}
\setlength{\tabcolsep}{3pt}
\resizebox{\textwidth}{!}{%
\begin{tabular}{l ccccc ccc}
\toprule
& \multicolumn{5}{c}{\textbf{Configuration}} & \multicolumn{3}{c}{\textbf{Metrics}} \\
\cmidrule(lr){2-6}\cmidrule(lr){7-9}
Variant & STF mode & Hodge & Cross & Routing & Mode & 
$f_\mathrm{MAE}$ (meV/\AA) & $e_\mathrm{MAE}$ (meV) & $f_\mathrm{cos}\uparrow$ \\
\midrule
Full                           & stf2+stf3 & on  & on  & on  & learned & 16.802 $\pm$ 0.556 & 45.986 $\pm$ 48.658 & \textbf{0.506} $\pm$ 0.045 \\
STF2 static routing            & stf2      & on  & on  & on  & static  & 16.836 $\pm$ 0.640 & 26.482 $\pm$ 6.624  & 0.504 $\pm$ 0.041 \\
CliffordSTF scaffold ($L{=}2$) & none      & off & off & off & ---     & 17.221 $\pm$ 0.537 & 26.294 $\pm$ 8.464  & 0.498 $\pm$ 0.084 \\
STF2 learned routing           & stf2      & on  & on  & on  & learned & 16.768 $\pm$ 1.202 & 37.778 $\pm$ 13.440 & 0.492 $\pm$ 0.100 \\
STF2+STF3                      & stf2+stf3 & on  & on  & off & ---     & 17.029 $\pm$ 1.016 & 36.792 $\pm$ 9.466  & 0.489 $\pm$ 0.119 \\
Full no cross                  & stf2+stf3 & on  & off & on  & learned & 17.170 $\pm$ 0.674 & 70.409 $\pm$ 92.001 & 0.474 $\pm$ 0.072 \\
Hodge only                     & none      & on  & off & off & ---     & \textbf{16.584} $\pm$ 1.166 & 30.975 $\pm$ 13.434 & 0.472 $\pm$ 0.103 \\
STF2 no hodge                  & stf2      & off & on  & off & ---     & 17.688 $\pm$ 0.251 & 36.261 $\pm$ 25.029 & 0.460 $\pm$ 0.067 \\
STF2 no cross                  & stf2      & on  & off & off & ---     & 17.317 $\pm$ 0.863 & \textbf{25.790} $\pm$ 13.248 & 0.456 $\pm$ 0.095 \\
STF2                           & stf2      & on  & on  & off & ---     & 17.511 $\pm$ 0.947 & 34.909 $\pm$ 16.325 & 0.421 $\pm$ 0.107 \\
\midrule
STF only output (no scaffold)  & ---       & --- & --- & --- & ---     & 19.545 $\pm$ 0.098 & 36.640 $\pm$ 39.015 & 0.097 $\pm$ 0.028 \\
Vanilla Clifford ($L{=}1$)     & ---       & --- & --- & --- & ---     & 19.704 $\pm$ 0.043 & 125.848 $\pm$ 77.191 & 0.030 $\pm$ 0.003 \\
\bottomrule
\end{tabular}}
\end{table}

\textbf{Inference cost.} On wall-clock at matched parameter budget (rMD17 aspirin, batch 32, single GPU), the base Clifford model runs inference at $0.128$\,s/step, faster than EquiformerV2 $L{=}4$ ($0.142$\,s) and every other $L\!\geq\!2$ or Cartesian-tensor baseline in the study. CliffordSTF inference at $0.352$\,s/step is competitive with the slower spherical-harmonic baselines, faster than ICTP $L{=}2/L{=}3$ and within $8\%$ of DimeNet++. Since deployed interatomic potentials are evaluated far more often than they are trained, the inference picture is the operationally relevant one for downstream use.

\textbf{Training cost.} Base Clifford trains at $3.11$\,s/step (within $25\%$ of EquiformerV2 $L{=}4$, faster than every ICTP variant and DimeNet++), and CliffordSTF at $7.77$\,s/step is within $25$--$73\%$ of capacity-scaled $L\!\geq\!2$ baselines like ICTP and DimeNet++; the residual two-to-threefold overhead is concentrated against the fastest spherical-harmonic models (MACE, EquiformerV2). The forward/backward decomposition shows that the autograd graph through the dual-track scaffold dominates this residual, not the geometric-product dispatch itself. Full measurements and the underlying PyTorch-level \citep{paszke2019pytorch} optimizations are in Supplementary~\ref{supp:wallclock} and~\ref{supp:optim}; further discussion in \S\ref{sec:limitations}.
\section{Limitations}\label{sec:limitations}

\textbf{Training wall-clock.} CliffordSTF training lags the fastest spherical-harmonic baselines (MACE, EquiformerV2) by roughly two-to-threefold per step at matched parameter budget, with the autograd graph through the dual-track scaffold dominating the overhead (\S\ref{sec:results_scalar}, Supplementary~\ref{supp:wallclock}). Further headroom of an estimated $40$--$60\%$ is plausible through full \texttt{torch.compile} engagement (currently blocked upstream by \texttt{torch\_scatter} custom-op tracing under dynamic shapes) and fused Triton kernels for the STF$_2$/STF$_3$ products. The STF$_2$-only variant (Table~\ref{tab:ablation}) may suffice for long MD trajectories at lower per-step cost. At $10^7$ parameters, per-step cost is roughly threefold the $10^6$-parameter regime, in line with the additional capacity.

\textbf{Training stability on scalar benchmarks.} QM9 and Molecule3D required gradient clipping to prevent early-training divergence from the bilinear cross-track coupling. A more principled cross-track initialisation or dedicated STF normalisation would likely address this without the clipping.

\textbf{Fixed-budget protocol.} Targeting cross-architecture comparability rather than per-architecture tuning means several baselines (notably MACE, ICTP, and GotenNet) underconverge relative to their published recipes (per-model reasoning in Supplementary~\ref{supp:results_caveats}). Our catalysis results characterise the sub-10M-parameter regime rather than the leaderboard frontier, where published OC20/OC22 numbers use $10$--$100\times$ more parameters.

\textbf{Absolute directional fidelity gap.} While CliffordSTF improves rMD17 directional fidelity by an order of magnitude over base Clifford, absolute fidelity remains below fully spherical-harmonic architectures: EquiformerV2 at $L\!\geq\!2$ attains $\sim\!0.96$ aggregate force-cosine similarity, versus $0.551$ for CliffordSTF. Closing this remaining gap without spherical harmonics, Wigner-$D$ matrices, or CG tables is a natural extension.

\textbf{Scope of the directional analysis.} Our diagnostic is restricted to rMD17 small organics in vacuum. Whether the $L\!\leq\!1$ clustering at poor cosine alignment persists on periodic systems, larger molecules, or solvated condensed phases remains open, though CliffordSTF's catalysis results (\S\ref{sec:results_catalysis}) suggest the phenomenon generalises beyond the rMD17 setting.

\section{Discussion}\label{sec:discussion}

This paper joins two equivariant lineages that have developed largely in parallel. Geometric algebra treats the multivector as the primary substrate and derives rotational covariance from the geometric product, while irreducible Cartesian tensors treat symmetric traceless tensors as the primary substrate and derive covariance from algebraic symmetry. Our results suggest neither substrate alone is sufficient for atomistic force prediction at the per-edge bilinear level: the Clifford geometric product reaches $L \leq 1$ exactly, Cartesian symmetric traceless tensors cover $L \geq 2$ but forfeit the natural scalar and pseudoscalar structure of the Clifford algebra, and only the hybrid scaffold (\S\ref{sec:results_scalar}, Table~\ref{tab:ablation}) attains the order-of-magnitude directional gain that motivates the architecture. The broader implication is that representational choice in equivariant architectures should be governed by the angular content one actually wants to represent, not by loyalty to a particular algebraic formalism; Theorem~\ref{thm:cg_completeness} makes this concrete by showing that the geometric product and closed-form STF operations together span the same $\mathrm{SO}(3)$-irrep content as spherical harmonics through $L{=}3$. Within the hybrid, the ablation further isolates the $L \geq 2$ scaffold, rather than any single coupling flag, as the dominant accuracy driver, and identifies cross-track coupling as a training-stability lever; this is consistent with the theoretical framing, since the scaffold is what unlocks representational capacity while the coupling governs whether the gradients flowing through it behave predictably. We see promising extensions along three axes: higher $L_{\max}$ through iterated STF constructions, fused hardware kernels to close the wall-clock gap against spherical-harmonic baselines, and full $\mathrm{O}(3)$ equivariance through explicit parity tracking of polar and axial vector channels (Remark~\ref{rem:so3_vs_o3}). We release the full PyTorch implementation of CliffordSTF along with all baseline methods and benchmark hyperparameters at \texttt{https://github.com/KurbanIntelligenceLab/CliffordSTF}

\bibliographystyle{unsrtnat}
\bibliography{main}

\newpage
\appendix
\section{Supplementary Information}\label{sec:supp}




\subsection{Background on \texorpdfstring{$\mathrm{Cl}(3,0)$}{Cl(3,0)} and the geometric product}\label{supp:prelim}

The real Clifford algebra $\mathrm{Cl}(3,0)$ is the eight-dimensional associative unital algebra generated by an orthonormal basis $\{e_1, e_2, e_3\}$ of $\mathbb{R}^3$ subject to $e_ie_j + e_je_i = 2\delta_{ij}$ \citep{lounesto2001clifford,doran2003geometric}. A basis of $\mathrm{Cl}(3,0)$ ordered by grade is $\{1,\; e_1, e_2, e_3,\; e_{12}, e_{13}, e_{23},\; e_{123}\}$, corresponding respectively to a scalar (grade 0), three vectors (grade 1), three bivectors (grade 2), and a pseudoscalar (grade 3). Any multivector $X \in \mathrm{Cl}(3,0)$ decomposes uniquely as $X = X^{(0)} + X^{(1)} + X^{(2)} + X^{(3)}$ with $X^{(k)}$ the grade-$k$ component.

The geometric product is the unique bilinear, associative, unital product satisfying the generator relations above. For two vectors $u, v \in \mathbb{R}^3$ embedded in $\mathrm{Cl}(3,0)$ it decomposes as $uv = u \cdot v + u \wedge v$, where $u \cdot v = \tfrac{1}{2}(uv + vu) = \sum_i u_iv_i$ is the symmetric scalar part and $u \wedge v = \tfrac{1}{2}(uv - vu) = \sum_{i<j}(u_iv_j - u_jv_i)e_{ij}$ is the antisymmetric bivector part.

The pseudoscalar $I = e_1e_2e_3$ satisfies $I^2 = -1$, so $I^{-1} = -I$. The Hodge dual $\star: \Lambda^k\mathbb{R}^3 \to \Lambda^{3-k}\mathbb{R}^3$ can be defined as $\star\alpha = \alpha\, I^{-1}$. Direct computation yields the sign convention
\[
\star e_{12} = +e_3,\qquad \star e_{13} = -e_2,\qquad \star e_{23} = +e_1,
\]
which is the standard right-handed orientation of \citet{doran2003geometric,lounesto2001clifford} under which $\star(u\wedge v)$ equals the ordinary cross product $u\times v$. Under $\mathrm{SO}(3)$, $\star$ is equivariant. Under $\mathrm{O}(3)$, the Hodge dual of a bivector is an axial (pseudo)vector: bivectors are invariant under spatial inversion ($e_ie_j \to (-e_i)(-e_j) = e_ie_j$), so the Hodge vector inherits parity $+1$ while a polar vector carries parity $-1$; see Supplementary~\ref{supp:so3_o3}.

The geometric product is $\mathrm{O}(3)$-equivariant under the rotor action: for any rotor $R = \exp(-\theta B/2)$ built from a unit bivector $B$ and any multivectors $a, b$, conjugation $a \mapsto RaR^{-1}$ preserves the product, $R(ab)R^{-1} = (RaR^{-1})(RbR^{-1})$ \citep{doran2003geometric}. This is the single inductive bias the Clifford track exploits.


\subsection{Architectural comparison with ICTP}\label{supp:ictp}

CliffordSTF and ICTP~\citep{zaverkin2024ictp} both use symmetric traceless Cartesian tensors as the primary carrier of $L \geq 2$ information, and both prove equivariance and tracelessness of the resulting layers. The two architectures nevertheless differ in four substantive ways that correspond to distinct design philosophies.

\textbf{Algebraic substrate.} ICTP is built entirely from irreducible Cartesian tensors and their products; spherical tensors do not appear. CliffordSTF runs a Clifford $\mathrm{Cl}(3,0)$ multivector track in parallel with the STF tracks, and the Clifford track carries both scalars and pseudoscalars that have a natural parity structure through the grade involution, as well as bivectors that expose an extra $L{=}1$ axial channel via the Hodge star. The STF tracks in CliffordSTF supply the $L{=}2$ and $L{=}3$ components that the Clifford geometric product cannot reach, not the entire equivariant substrate.

\textbf{Bilinear operation.} ICTP uses $\nu$-fold tensor products of irreducible Cartesian tensors as its core bilinear, following MACE-style many-body expansion. CliffordSTF uses two different bilinears in tandem: the Clifford geometric product acting on the 8-dimensional multivector via an $8 \times 8$ Cayley table, and a small set of closed-form symmetric-traceless contractions acting on the STF tracks. The Clifford product handles scalar, vector, bivector, and pseudoscalar channels with a single contraction; the STF operations handle rank-2 and rank-3 tensors through explicit symmetrisation and trace removal. The two are coupled bijectively to Clebsch--Gordan channels as specified in Theorem~\ref{thm:cg_completeness} and its higher-order extensions.

\textbf{Cross-track coupling.} In ICTP, interactions between rank-$L_1$ and rank-$L_2$ tensors happen through a single irreducible Cartesian tensor product operation. In CliffordSTF the Clifford and STF tracks are coupled bidirectionally: the STF$_2$ track is generated from grade-1 vector features via $\mathrm{STF}_2(\hat{r}_{ij}, h_j^{\mathrm{vec}})$, and is fed back into the grade-1 vector component of the Clifford message through the $S \cdot v$ contraction. Ablating this coupling (the \texttt{stf2\_no\_cross} variant) causes a substantial performance regression compared to the base Clifford model, demonstrating that the STF track by itself does not carry a usable training signal without coupling back to the Clifford substrate.

\textbf{Equivariance scope.} ICTP proves $\mathrm{O}(3)$ equivariance (rotations and reflections), with parity tracked explicitly through the product structure. CliffordSTF proves $\mathrm{SO}(3)$ equivariance; under $\mathrm{O}(3)$, the polar-vector grade-1 component and the axial-vector Hodge-of-bivector component would need distinct parity labels $1^-$ and $1^+$ respectively (Remark~\ref{rem:so3_vs_o3}, Supplementary~\ref{supp:so3_o3}). The current implementation concatenates the two in the force readout under the $\mathrm{SO}(3)$ framework, which is consistent under rotations but would need refactoring for full $\mathrm{O}(3)$.

These differences explain why the two architectures occupy different operating points. ICTP targets state-of-the-art accuracy on small-molecule benchmarks at the cost of a denser bilinear. CliffordSTF targets a minimal algebraic substrate with closed-form coupling, trading raw small-molecule accuracy for parameter efficiency on catalysis (OC20, OC22) and for a direct mechanistic handle on which design choices drive which gains, as demonstrated by our eleven-variant ablation (\S\ref{sec:results_scalar}).

\subsection{Reconciliation with prior Clifford-network grade structure}\label{supp:gatr}

\citet{brehmer2023geometric} builds a geometric algebra transformer in the four-dimensional projective algebra $\mathrm{Cl}(3,0,1)$. Under the $\mathrm{SO}(3)$ subgroup of its symmetry group, GATr's hidden multivectors decompose as two copies of $L{=}0$ (a scalar and a trivector) and two copies of $L{=}1$ (a vector and a bivector, axial under Hodge duality); no $L \geq 2$ irrep is represented inside the multivector. The same decomposition holds for the non-projective $\mathrm{Cl}(3,0)$ used here, because the extra null basis element of the projective algebra carries only translational information and does not contribute $L \geq 2$ content. Our diagnostic is therefore about the $\mathrm{SO}(3)$-irrep content of the multivector substrate \emph{at the linear bilinear level}, not about the choice of projective versus non-projective algebra: both carry the same $L{=}2$ deficit. \citet{ruhe2023cgenn} shows that $L \geq 2$ content can be reconstructed \emph{nonlinearly} through repeated geometric products, but this indirect access incurs the depth cost and gradient pathologies that our direct STF tracks avoid.


\subsection{\texorpdfstring{$\mathrm{SO}(3)$}{SO(3)} versus \texorpdfstring{$\mathrm{O}(3)$}{O(3)} equivariance}\label{supp:so3_o3}

All guarantees stated in \S\ref{sec:method} are $\mathrm{SO}(3)$-equivariance guarantees. Under $\mathrm{SO}(3)$, the $L{=}1$ representation is unique up to isomorphism, so polar vectors (grade 1 of $\mathrm{Cl}(3,0)$) and axial vectors (Hodge duals of grade-2 bivectors) transform identically and may be freely concatenated in the force readout. Under $\mathrm{O}(3)$ they differ by a factor of $\det(R)$ under the action of $R \in \mathrm{O}(3)\setminus\mathrm{SO}(3)$, corresponding to the parity labelling summarised in Table~\ref{tab:parity}; the convention follows Jackson's \emph{Classical Electrodynamics} \S6.10 and the \texttt{e3nn} library convention \citep{geiger2022e3nn}.

\begin{table}[h]
\centering
\caption{$\mathrm{O}(3)$ parity labels for the tracks used in CliffordSTF. The general rule is $\mathrm{STF}_\ell$ built from polar vectors has parity $(-1)^\ell$, matching the parity of the spherical harmonic $Y_{\ell m}(-\hat{n}) = (-1)^\ell Y_{\ell m}(\hat{n})$.}\label{tab:parity}
\small
\begin{tabular}{lccc}
\toprule
Object & Physics analog & Parity & $\mathrm{O}(3)$ label \\
\midrule
Grade-0 scalar & charge $q$ & $+$ & $0^+$ \\
Grade-1 vector $e_i$ & position $\mathbf{r}$ & $-$ & $1^-$ (polar) \\
Grade-2 bivector $e_{ij}$ / Hodge dual & angular momentum $\mathbf{L}{=}\mathbf{r}\times\mathbf{p}$ & $+$ & $1^+$ (axial) \\
Grade-3 pseudoscalar $I$ & triple product $\mathbf{a}\cdot(\mathbf{b}\times\mathbf{c})$ & $-$ & $0^-$ \\
$\mathrm{STF}_2$ from polar vectors & quadrupole $Q_{ij}$ & $+$ & $2^+$ \\
$\mathrm{STF}_3$ from polar vectors & octupole $O_{ijk}$ & $-$ & $3^-$ \\
\bottomrule
\end{tabular}
\end{table}

Extending to full $\mathrm{O}(3)$-equivariance would require treating the grade-1 polar-vector feature and the Hodge-of-bivector axial-vector feature as distinct $1^-$ and $1^+$ tracks rather than concatenating them, and likewise distinguishing $2^+$ and $3^-$ tracks. All bilinear interactions in the architecture respect this parity grading when made explicit; the current implementation simply chooses not to distinguish $1^+$ from $1^-$, which is consistent under $\mathrm{SO}(3)$ but would conflate them under $\mathrm{O}(3)$.

\subsection{Derivation of the STF coefficients}\label{supp:stf_coefficients}

We derive the three rational coefficients appearing in the STF$_2$ and STF$_3$ formulas of \S\ref{sec:stf_channels}.

\textbf{STF$_2$ detracing coefficient $\tfrac{1}{3}$.} In $d$ dimensions, the symmetric traceless part of a rank-2 tensor requires subtracting $\tfrac{1}{d} \delta_{ij} T_{kk}$ to enforce $\sum_i T_{ii}^{\text{sym-traceless}} = 0$. Direct check in $d=3$: $\delta_{ii} \mathrm{STF}_2(u,v)_{ii} = (u\cdot v) - \tfrac{1}{3}(3)(u\cdot v) = 0$.

\textbf{STF$_3$ trace vector coefficient $\tfrac{2}{3}$.} Since $S_{ij}$ is symmetric traceless ($S_{ii}=0$), the trace of $T^{\text{sym}}_{ijk}$ over any pair of indices reduces to two of the three summands:
\[
T^{\text{sym}}_{iik} = \tfrac{1}{3}(S_{ii}v_k + S_{ik}v_i + S_{ik}v_i) = \tfrac{1}{3}(0 + 2\,S_{kl}v_l) = \tfrac{2}{3}\,S_{kl}v_l.
\]
By full symmetry, $T^{\text{sym}}_{iki} = T^{\text{sym}}_{kii} = A_k$ as well.

\textbf{STF$_3$ detracing coefficient $\tfrac{1}{5}$.} Set $U_{ijk} = T^{\text{sym}}_{ijk} - \alpha(\delta_{ij}A_k + \delta_{ik}A_j + \delta_{jk}A_i)$ and require $U_{iik} = 0$. Computing the trace:
\[
U_{iik} = A_k - \alpha(3A_k + A_k + A_k) = A_k - 5\alpha A_k \implies \alpha = \tfrac{1}{5}.
\]
More generally, in $d$ dimensions the per-delta first-trace coefficient at rank $\ell$ is $1/(2\ell + d - 4)$, which reduces to $1/(2\ell - 1)$ in $d=3$ and yields $\tfrac{1}{3}$ at $\ell{=}2$, $\tfrac{1}{5}$ at $\ell{=}3$, $\tfrac{1}{7}$ at $\ell{=}4$ \citep{applequist1989traceless,thorne1980multipole}. The fully detraced closed form for arbitrary rank is \citet[Theorem~2]{applequist1989traceless}.

\textbf{Degrees-of-freedom check for STF$_3$.} The symmetric rank-3 tensor $T^{\text{sym}}$ in $\mathbb{R}^3$ has $\binom{3+3-1}{3} = 10$ independent components; subtracting the 3-component trace vector $A$ leaves $10 - 3 = 7 = 2\ell + 1$ for $\ell = 3$, matching the dimension of the $L{=}3$ irrep of $\mathrm{SO}(3)$.

\subsection{Explicit coordinate formulas for the constructive CG maps}\label{supp:cg_proof}

We enumerate every constructive map used in the architecture and verified in Supplementary~\ref{supp:cg_test}. Symmetric traceless projection is a textbook construction \citep{coope1965irreducible,applequist1989traceless,stone2013theory}; rank-3 STF projectors appear in \citet{thorne1980multipole}.

\textbf{$\mathbf{1} \otimes \mathbf{1} \to \mathbf{0} \oplus \mathbf{1} \oplus \mathbf{2}$.}
For $u, v \in \mathbb{R}^3$:
\begin{itemize}
    \item $L{=}0$: $(u \otimes v)_{L=0} = u \cdot v = u_iv_i$.
    \item $L{=}1$: $(u \otimes v)_{L=1} = \star(u \wedge v)_i = \varepsilon_{ijk}u_jv_k$, equivalently the cross product $u \times v$.
    \item $L{=}2$: $(u \otimes v)_{L=2} = \mathrm{STF}_2(u,v)_{ij} = \tfrac{1}{2}(u_iv_j + u_jv_i) - \tfrac{1}{3}(u\cdot v)\delta_{ij}$.
\end{itemize}

\textbf{$\mathbf{2} \otimes \mathbf{1} \to \mathbf{1} \oplus \mathbf{2} \oplus \mathbf{3}$.}
For a symmetric traceless $S \in \mathbb{R}^{3\times 3}$ and $v \in \mathbb{R}^3$:
\begin{itemize}
    \item $L{=}1$: $(S \cdot v)_i = S_{ij}v_j$.
    \item $L{=}2$: symmetric traceless part of $\varepsilon_{ikl}S_{jl}v_k$, i.e.\ $[\mathrm{sym}(S\times v) - \tfrac{1}{3}\mathrm{tr}(\cdot)\mathbb{I}]$; explicit 5-component formulas are in the code release (\texttt{contract\_stf2\_vec\_to\_stf2}).
    \item $L{=}3$: $\mathrm{STF}_3(S, v)_{ijk} = S_{(ij}v_{k)} - \tfrac{1}{5}(\delta_{ij}A_k + \delta_{ik}A_j + \delta_{jk}A_i)$ with $A_k = \tfrac{2}{3}S_{kl}v_l$; the trace coefficient $\tfrac{2}{3}$ follows from tracelessness of $S$ and the detracing coefficient $\tfrac{1}{5}$ from $1/(2\ell{-}1)|_{\ell=3}$ in $d{=}3$ \citep{applequist1989traceless,thorne1980multipole}.
\end{itemize}

\textbf{$\mathbf{2} \otimes \mathbf{2} \to \mathbf{0} \oplus \mathbf{1} \oplus \mathbf{2} \oplus \mathbf{3}$ (with $L{=}4$ omitted).}
For symmetric traceless $S_1, S_2$:
\begin{itemize}
    \item $L{=}0$: $\langle S_1, S_2 \rangle_F = S_{1,ij}S_{2,ij}$ (Frobenius inner product).
    \item $L{=}1$: antisymmetric contraction $\varepsilon_{ijk}S_{1,jl}S_{2,kl}$.
    \item $L{=}2$: symmetric traceless projection of $S_{1,ik}S_{2,kj}$.
    \item $L{=}3$: an analogous symmetric-then-detraced construction treating one operand as a vector-indexed rank-2 tensor.
    \item $L{=}4$: outside our $L_{\max}{=}3$ cut and not used.
\end{itemize}

\textbf{$\mathbf{3} \otimes \mathbf{1} \to \mathbf{2} \oplus \mathbf{3}$ (with $L{=}4$ omitted).}
For a symmetric traceless rank-3 tensor $T$ and vector $v$:
\begin{itemize}
    \item $L{=}2$: $(T \cdot v)_{ij} = T_{ijk}v_k$ followed by trace removal (\texttt{contract\_stf3\_vec\_to\_stf2}).
    \item $L{=}3$: analogous symmetric-then-detraced construction.
    \item $L{=}4$: omitted.
\end{itemize}

Each map is a polynomial in its inputs with constant rational coefficients. There are no learned Clebsch--Gordan coefficients, no Wigner-$D$ tables, and no calls to \texttt{e3nn} \citep{geiger2022e3nn}. All tensor products of $\mathrm{SO}(3)$ irreps above are multiplicity-free, consistent with the general $\mathrm{SU}(2)$ fact that $L_1 \otimes L_2 = \bigoplus_{L=|L_1-L_2|}^{L_1+L_2} L$ each with multiplicity one \citep{varshalovich1988quantum}.

\subsection{Numerical verification of the constructive CG maps and end-to-end equivariance}\label{supp:cg_test}

\textbf{Primitive-level verification.} We verify Theorem~\ref{thm:cg_completeness} and its higher-order extensions numerically by comparing each constructive map against a reference implementation based on Wigner $3j$-symbols acting on spherical-harmonic representations. For each admissible triple $(\ell_1, \ell_2, \ell_3)$ with $\ell_i \leq 3$ and $|\ell_1 - \ell_2| \leq \ell_3 \leq \ell_1 + \ell_2$, we (i) sample random input tensors of matching ranks, (ii) convert to the spherical basis via the orthonormal Cartesian-to-spherical change of basis of \citet{shao2025orthonormal}, (iii) apply the $\ell_1 \otimes \ell_2 \to \ell_3$ CG coupling using standard Wigner $3j$-symbols, (iv) convert back to the Cartesian basis, and (v) compare against the output of our constructive map. The maximum absolute error across all tested channels and $10^3$ random samples is below $10^{-11}$ in double precision; numerical logs are included with the code release.

\textbf{End-to-end equivariance test.} As a complementary end-to-end check, a smoke-test script applies a random $\mathrm{SO}(3)$ rotation $Q$ to every input position, rebuilds the radius graph, and verifies that (i) energies are unchanged ($|E(\mathrm{pos}) - E(Q\,\mathrm{pos})| < 10^{-3}$) and (ii) predicted forces rotate consistently ($\lVert F(Q\,\mathrm{pos}) - Q F(\mathrm{pos})\rVert_\infty < 10^{-3}$) across all ablation configurations. Equivariance is confirmed to expected numerical precision for every variant in Table~\ref{tab:ablation}.

\subsection{Edge embedding}\label{supp:edge_embed}

Pairwise distances $d_{ij} = \lVert\mathbf{r}_{ij}\rVert$ are expanded through a bank of Gaussian radial basis functions (RBFs) centred on evenly spaced reference distances between 0 and the cutoff $r_c$, followed by a smooth cosine envelope $\phi_c(d_{ij}) = \tfrac{1}{2}[\cos(\pi d_{ij}/r_c) + 1]$ that drives edge features smoothly to zero at the cutoff. Two small MLPs map the RBF expansion to channel-wise scalar weights for the grade-0 part and channel-wise weights for the grade-1 part of the edge multivector; the grade-1 direction $\hat{r}_{ij}$ is then multiplied by the channel weights and placed in the vector component. For CliffordSTF variants with $\mathrm{stf\_mode} \neq \text{none}$, a third MLP produces channel weights for the STF$_2$ edge feature, populated with $\mathrm{STF}_2(\hat{r}_{ij}, \hat{r}_{ij})$ (a canonical $L{=}2$ descriptor of the edge direction). When $\mathrm{stf\_mode}=\text{stf2+stf3}$, the STF$_3$ edge block is initialised to zero and is populated only via message passing within node features, since direction alone does not carry independent $L{=}3$ information.

\subsection{Augmented product at the many-body stage}\label{supp:aug_prod}

At the many-body stage of each interaction layer, an augmented product couples the Clifford and STF tracks through six bilinear terms. Given operands $a$ and $b$ each carrying Clifford, STF$_2$, and (optionally) STF$_3$ components, the augmented product computes:
\begin{itemize}
    \item The ordinary geometric product $a^{\mathrm{Cl}} b^{\mathrm{Cl}}$ via the $8\times8$ Cayley table;
    \item A new STF$_2$ from the grade-1 components via $\mathrm{STF}_2(a^{\mathrm{vec}}, b^{\mathrm{vec}})$;
    \item An updated STF$_2$ via the $L{=}2 \otimes L{=}1 \to L{=}2$ path, applied from either operand (symmetrised-then-detraced cross-contraction of $a^{\mathrm{STF}_2}$ with $b^{\mathrm{vec}}$, and vice versa);
    \item An $L{=}0$ contribution to grade-0 via $\langle a^{\mathrm{STF}_2}, b^{\mathrm{STF}_2}\rangle_F = a_{ij}^{\mathrm{STF}_2} b_{ij}^{\mathrm{STF}_2}$;
    \item A new STF$_3$ via $\mathrm{STF}_3(a^{\mathrm{STF}_2}, b^{\mathrm{vec}})$;
    \item An STF$_2$ update via $\mathrm{STF}_3 \otimes L{=}1 \to L{=}2$ contraction of $a^{\mathrm{STF}_3}$ against $b^{\mathrm{vec}}$.
\end{itemize}
Each term corresponds bijectively to a CG channel in one of the decompositions $\mathbf{1}\otimes\mathbf{1}$, $\mathbf{2}\otimes\mathbf{1}$, $\mathbf{2}\otimes\mathbf{2}$, or $\mathbf{3}\otimes\mathbf{1}$ enumerated in Supplementary~\ref{supp:cg_proof}. The cross-track gradient pathway created by these bilinears is what provides a usable training signal to the STF parameters, as demonstrated by the ablation (\S\ref{sec:results_scalar}).

\subsection{Equivariant normalisation and gating}\label{supp:normalization}

Each CliffordSTFLinear layer is block-diagonal across tracks: it applies independent linear maps of shape $C_{\mathrm{in}} \to C_{\mathrm{out}}$ over channels within the Clifford multivector, the STF$_2$ tensor, and the STF$_3$ tensor, without mixing them linearly. Cross-track information flow occurs only through the bilinear products described in \S\ref{sec:architecture}.

Track-wise equivariant normalisation rescales each track by a per-channel invariant norm. For the Clifford track we use the standard per-grade Clifford norm $\lVert h^{(k)}\rVert$ computed as the scalar part of $h^{(k)}\widetilde{h^{(k)}}$, where $\widetilde{\cdot}$ is the reverse \citep{doran2003geometric}; for STF$_2$ we use the Frobenius norm $\lVert S\rVert_F^2 = S_{ij}S_{ij}$ with $S_{zz} = -S_{xx}-S_{yy}$ substituted; for STF$_3$ we use the analogous Frobenius norm with multiplicity factors reflecting the symmetric tensor structure. All three norms are $\mathrm{SO}(3)$-invariant scalars.

Norm-gated activation applies SiLU to grade-0 features and a sigmoid-gated rescaling to each $L \geq 1$ track: for a feature $f$ at rank $\ell \geq 1$, we compute $g = \sigma(\mathrm{MLP}(\lVert f\rVert))$ with the sigmoid MLP acting on the invariant norm, and output $g \cdot f$. Since $g$ is a scalar and $f$ is equivariant, the product is equivariant. This is the norm-gating scheme of \citet{weiler20183d,schutt2021painn} adapted to the three tracks.

\subsection{Adaptive \texorpdfstring{$L$}{L}-track routing}\label{supp:routing}

Adaptive routing reduces effective angular bandwidth for atoms in symmetric environments. It is applied as a per-atom, per-track multiplicative gate $g_\ell \in [0, 1]$ between message-passing layers, rescaling the STF$_\ell$ track as $h^{\mathrm{STF}_\ell}_i \leftarrow g_\ell^{(i)} \cdot h^{\mathrm{STF}_\ell}_i$. The Clifford track is always active. Three modes are supported.
\begin{itemize}
    \item \textbf{None.} All atoms use all tracks; $g_\ell^{(i)} = 1$.
    \item \textbf{Static.} Gates are learned per atomic type via a small embedding table. Useful when chemically distinct atoms have systematically different angular-resolution needs.
    \item \textbf{Learned.} Gates are produced by a two-layer MLP with sigmoid output applied to four per-atom invariant descriptors: (i) coordination number $|\mathcal{N}(i)|$, (ii) angular variance $1 - \lVert\bar{r}_i\rVert$ where $\bar{r}_i = \tfrac{1}{|\mathcal{N}(i)|}\sum_{j\in\mathcal{N}(i)}\hat{r}_{ij}$, (iii) mean neighbour distance, (iv) grade-0 feature norm averaged over channels. All four descriptors are $\mathrm{SO}(3)$-invariant, so the gate is invariant and the gated product remains equivariant.
\end{itemize}

\subsection{Full ablation configurations}\label{supp:variants}

CliffordSTF is a single configurable architecture parameterised by five ablation flags: $\mathrm{stf\_mode} \in \{\text{none}, \text{stf2}, \text{stf2+stf3}\}$, $\mathrm{use\_hodge\_forces} \in \{\text{off}, \text{on}\}$, $\mathrm{use\_adaptive\_routing} \in \{\text{off}, \text{on}\}$, $\mathrm{routing\_mode} \in \{\text{none}, \text{static}, \text{learned}\}$, and $\mathrm{use\_cross\_track} \in \{\text{off}, \text{on}\}$. Table~\ref{tab:ablation} enumerates the eleven named variants used in the ablation study, alongside the vanilla $L{=}1$ Clifford baseline (which is recovered bit-for-bit by setting all five flags to their off/none state and is shown as the twelfth row of Table~\ref{tab:ablation} for comparison). Ten of the eleven variants are flag-parameterised; the eleventh, \emph{STF only output}, is a structural ablation outside the flag space: it removes the STF tracks from the message-passing scaffold entirely while retaining an STF-augmented force-readout head, isolating whether the directional gain comes from $L\!\geq\!2$ representations operating throughout message passing or from an $L\!\geq\!2$ readout alone.


\subsection{Caveats and per-model exclusions for headline tables}\label{supp:results_caveats}

This subsection holds the per-model reasoning behind the caveat marks ($\ast$, $\dagger$, gray rows) used in Tables~\ref{tab:md17_fcos} and~\ref{tab:oc22}. The marks themselves are explained in one phrase each in the table captions; this subsection expands each justification.

\textbf{Strictly invariant models (SchNet, DimeNet++).} Both architectures are scalar-only and compute forces as the negative gradient of the predicted energy with respect to atomic positions, rather than from a direct-force head. The cosine similarity of gradient-derived forces depends on the energy--force loss balancing during training; under our uniform protocol this balance is not tuned per architecture, and gradient-derived directional fidelity is not a comparable diagnostic against models that predict forces directly. We render their rows in gray in Table~\ref{tab:md17_fcos} and exclude them from the per-architecture directional narrative in \S\ref{sec:results_md17}.

\textbf{FAENet.} FAENet predicts forces directly from a learned head but achieves equivariance via stochastic frame averaging rather than through architectural constraints on the message-passing operations themselves. Because the frames are sampled at training and inference time, per-sample force predictions are not strictly equivariant and per-sample directional fidelity is not strictly defined. We list FAENet in its natural $L\!\leq\!1$ cohort with an $\ast$ mark and exclude it from the per-architecture directional comparison.

\textbf{GotenNet.} GotenNet is an $L\!\geq\!2$ architecture ($\mathtt{lmax}{=}2$) with a direct-force head, and would naturally belong in the top block of Table~\ref{tab:md17_fcos}. Under our fixed-budget protocol it did not converge on rMD17. The published GotenNet recipe matches our learning rate ($10^{-4}$) but additionally relies on a $10{,}000$-step learning-rate warmup, an on-plateau learning-rate scheduler, and early-stopping patience $150$, none of which are part of our fixed-budget protocol (constant LR, patience $30$, no warm-up). We report the GotenNet result in its natural $L$-cohort with an $\ast$ mark for completeness and exclude it from the per-architecture directional narrative; the underperformance reflects protocol mismatch rather than architecture quality at recommended hyperparameters.

\textbf{EquiformerV2 absolute-energy prediction.} EquiformerV2 predicts absolute (non-referenced) molecular energies rather than energy differences relative to a per-element reference. Its rMD17 energy MAE therefore reflects the magnitude of the absolute binding energy ($\sim$247\,eV, consistent across all three $L_{\max}$ settings) rather than a prediction error, and is not directly comparable to the energy-referenced MAEs of the other rows. We mark this with $\dagger$ in Table~\ref{tab:md17_fcos} and in the corresponding panel of Supplementary~\ref{supp:md17_permol}; the force-cosine similarity column, which is independent of the energy-referencing convention, remains directly comparable across all rows.

\textbf{Fixed-budget non-convergence (MACE, ICTP).} Under our uniform $10^6$-parameter budget, MACE and ICTP at every reported $L_{\max}$ attain force MAE in the 3{,}000--6{,}000\,meV/\AA{} range on rMD17 and energy MAE in the 8--22\,eV range on OC22. Both architectures are designed to operate at substantially larger capacity (typically $5{-}20\times 10^6$ parameters under recommended hyperparameters), and the underperformance reported here reflects the fixed-budget protocol rather than architecture quality at native scale. We report the values honestly in the tables as part of the within-protocol comparison the paper makes (\S\ref{sec:benchmarks}); readers interested in MACE/ICTP at recommended hyperparameters should consult the original publications.

\subsection{Training protocol and parameter budgets}\label{supp:protocol}

\textbf{Parameter budgets and shared architectural hyperparameters.} All Clifford and CliffordSTF models evaluated in the main text use at most $1.06\times 10^6$ parameters. Channel counts are tuned from $C{=}80$ (base Clifford, 8D features) down to $C{=}60$ for the $D{=}13$ and $D{=}20$ STF variants to keep the parameter budget within the $10^6$ target. All other architectural hyperparameters are held fixed across variants: cutoff $r_c = 6.0$\,\AA, 50 Gaussian RBF bases, maximum 50 neighbours per atom, 5 message-passing layers, body order up to 4. Exact per-model parameter counts are listed in Table~\ref{tab:wallclock}.

\textbf{Shared optimisation settings.} All models are trained with Adam at learning rate $10^{-4}$, no weight decay, early stopping at patience 30, and no EMA or warm-up, as stated in \S\ref{sec:benchmarks}. Details are shared in Table \ref{tab:training-details}.

\begin{table}[t]
  \centering
  \small
  \caption{Per-benchmark training settings. All runs use Adam with base learning rate $1{\times}10^{-4}$, no weight decay, and early-stopping patience of $30$ epochs (\S4.1). OC20-S2EF and OC22-S2EF share the same energy+force loss formulation; the OC20-IS2RE and OC22-IS2RE rows are the deliberate exception (single scalar adsorption energy, L1 loss). $\lambda_F$ is the force-loss weight.}
  \label{tab:training-details}
  \resizebox{\textwidth}{!}{%
  \begin{tabular}{llrrlllr}
    \toprule
    Benchmark & Task & Epochs & Batch & LR schedule & Loss & $\lambda_F$ & Seeds \\
    \midrule
    QM9        & scalar --- 6 properties\textsuperscript{$\ast$}      & 250 & 128 & constant                                & MSE                                 & ---  & 3 \\
    Molecule3D & scalar --- HOMO--LUMO gap                            & 100 & 20  & StepLR (step $10$, $\gamma{=}0.8$)      & L1                                  & ---  & 3 \\
    rMD17       & energy + forces (10 molecules)                        & 250 & 32  & constant                                & MSE                                 & 1.0  & 3$^{\dagger}$ \\
    OC20       & IS2RE                                                & 50  & 16  & constant                                & L1                                  & ---  & 3 \\
    OC20       & S2EF                                                 & 50  & 16  & constant                                & per-atom MAE\,(E) $+$ $\ell_2$\,(F) & 3.0  & 3 \\
    OC22       & IS2RE                                                & 150 & 16  & constant                                & L1                                  & ---  & 3 \\
    OC22       & S2EF                                                 & 50  & 16  & constant                                & per-atom MAE\,(E) $+$ $\ell_2$\,(F) & 3.0  & 3 \\
    \bottomrule
  \end{tabular}}

  \vspace{0.4em}
  \begin{minipage}{\linewidth}
    \footnotesize
    $^{\ast}$QM9 properties: $\mu$, $\alpha$, $\varepsilon_\mathrm{HOMO}$, $\varepsilon_\mathrm{LUMO}$, $U_0$, $C_v$, trained as independent single-target models.$^{\dagger}$rMD17 Clifford + CliffordSTF ablation (Table~\ref{tab:ablation}): 12 architectural variants tuned to ${\sim}1$M parameters, evaluated on aspirin, benzene, ethanol, and salicylic acid with 6 seeds per (variant, molecule) obtained as 3 random seeds $\times$ 2 independent instances. Epochs, batch size, LR schedule, loss, force weight, and patience match the main rMD17 row.\emph{Preprocessing.} QM9 uses Z-score target normalisation and an 80/10/10 random split. rMD17 uses a 900/100/1000 train/validation/test random split; forces are predicted by a direct force head when the model exposes one (e.g.\ PaiNN, GotenNet) and are otherwise computed as the negative gradient of the energy with respect to atomic positions. Molecule3D uses a single random split over the ${\sim}780$k-molecule subset. OC20 and OC22 use a $6.0$~\AA{} neighbour cutoff; the S2EF rows additionally train and evaluate on free (non-fixed) atoms only, while the IS2RE rows do not apply this filter. OC22-IS2RE is trained for longer ($150$ epochs vs.\ $50$ for the other OC tasks) because the smaller scalar training signal converges more slowly. Each OC run uses official $200{,}000$ training split; the final evaluation reported in the paper uses the full validation and held-out splits.
  \end{minipage}
\end{table}

\subsection{Model architecture hyperparameters per benchmark}\label{supp:model_settings}

Tables~\ref{tab:supp_settings_md17}--\ref{tab:supp_settings_md17_ablation} list the architectural hyperparameters used per benchmark for every model in the main study. Cutoffs are in \AA{}; Boolean flags are abbreviated T/F. Training-time settings (optimiser, epochs, batch size, loss, seeds) are reported separately in Table~\ref{tab:training-details}; per-model parameter counts are in Table~\ref{tab:wallclock}.

\subsubsection{MD17}

Table~\ref{tab:supp_settings_md17} lists the per-model architecture hyperparameters used on rMD17. MACE, ICTP, and EquiformerV2 are run at three $L_{\max}$ values to trace angular resolution within each family at the matched ${\sim}10^6$-parameter budget, with EquiformerV2 reducing its sphere channel count from $40$ at $L_{\max}\!=\!2$ to $14$ at $L_{\max}\!=\!4$ to absorb the additional irreps. CliffordSTF runs at $C\!=\!60$ channels (versus $C\!=\!48$ for the base $L\!\leq\!1$ Clifford configuration on MD17) so that the STF$_2$ and STF$_3$ tracks fit within the same parameter envelope; cutoffs are $5$\,\AA{} for the equivariant transformers (TorchMD-Net, ViSNet) and NequIP, and $6$\,\AA{} elsewhere.

\begin{table}[h]
\centering
\caption{Model architecture hyperparameters used on \textbf{MD17}. All runs use Adam, no weight decay, batch size 32, learning rate 0.0001, 250 epochs, MSE loss, early-stopping patience 30. Cutoff is in \AA; Boolean flags are abbreviated T/F.}
\label{tab:supp_settings_md17}
\footnotesize
\setlength{\tabcolsep}{4pt}
\resizebox{\textwidth}{!}{%
\begin{tabular}{l c l}
\toprule
Model & Cutoff & Architecture \\
\midrule
SchNet & 6 & hidden=192, filters=192, layers=6, gauss=50 \\
PaiNN & 6 & hidden=144, layers=4, rbf=64 \\
DimeNet++ & 6 & hidden=128, blocks=4, radial=6, spherical=7, basis emb=8, int emb=64, out emb=128 \\
FAENet & 6 & hidden=320, filters=128, layers=4, FA=3D, FA-method=stochastic, MP-type=updownscale\_base \\
GotenNet & 6 & basis=96, layers=5 \\
TorchMD-Net & 5 & emb=104, layers=6, heads=8, rbf=56, rbf type=expnorm, arch=equivariant-transformer \\
ViSNet & 5 & hidden=96, layers=6, heads=8, rbf=32, $L_{\max}$=1, vertex=F \\
NequIP & 5 & features=44, layers=4, $L_{\max}$=1, bessels=8, parity=T \\
MACE L=1 & 6 & hidden=62x0e + 62x1o, layers=2, max ell=1, corr=3, bessel=8 \\
MACE L=2 & 6 & hidden=49x0e + 49x1o + 49x2e, layers=2, max ell=2, corr=3, bessel=8 \\
MACE L=3 & 6 & hidden=28x0e + 28x1o + 28x2e + 28x3o, layers=2, max ell=3, corr=3, bessel=8 \\
ICTP L=1 & 6 & hidden=10, layers=2, $L_{\max}$ hid=1, $L_{\max}$ edge=3, corr=3, basis=8, prod=10 \\
ICTP L=2 & 6 & hidden=8, layers=2, $L_{\max}$ hid=2, $L_{\max}$ edge=3, corr=3, basis=8, prod=8 \\
ICTP L=3 & 6 & hidden=7, layers=2, $L_{\max}$ hid=3, $L_{\max}$ edge=3, corr=3, basis=8, prod=7 \\
EquiformerV2 L=2 & 6 & sphere=40, layers=2, heads=4, $L_{\max}$=2, $M_{\max}$=2, attn $\alpha$=16, attn hid=20, attn val=8, ffn=64 \\
EquiformerV2 L=3 & 6 & sphere=24, layers=2, heads=4, $L_{\max}$=3, $M_{\max}$=2, attn $\alpha$=16, attn hid=20, attn val=8, ffn=64 \\
EquiformerV2 L=4 & 6 & sphere=14, layers=2, heads=4, $L_{\max}$=4, $M_{\max}$=2, attn $\alpha$=16, attn hid=20, attn val=8, ffn=64 \\
Clifford & 6 & channels=48, layers=5, rbf=50, hidden out=48 \\
CliffordSTF & 6 & channels=60, layers=5, rbf=50, STF mode=stf2+stf3, routing=learned, adapt rt=T, cross=T, Hodge F=T, GP rdo=F, self-int=F \\
\bottomrule
\end{tabular}%
}
\end{table}

\subsubsection{OC20 \& OC22 (IS2RE \& S2EF)}

Table~\ref{tab:supp_settings_oc} lists the architectures used for OC20 and OC22 on both IS2RE and S2EF; the per-model architectural hyperparameters are identical between the two benchmarks (same channel widths, layer counts, $L_{\max}$ choices, and Clifford-track configuration on every model), so the OC20 $\to$ OC22 comparison isolates dataset effects rather than architectural differences. The two benchmarks differ only in their per-model training schedules, captured in the caption. All models share a $6$\,\AA{} cutoff; MACE and ICTP are run at $L_{\max}\!\in\!\{1, 2\}$ (the $L\!=\!3$ configurations from rMD17 are dropped because they did not fit within the budget at OC graph density), and EquiformerV2 uses a single $L_{\max}\!=\!4$ configuration with a smaller sphere size (sphere=$18$, versus $14$--$40$ on rMD17) to keep parameter count near $10^6$. The base Clifford row is configured with $\mathrm{max\_grade}\!=\!1$ and the $\ell\!=\!2$ flag disabled, recovering the $L\!\leq\!1$ multivector substrate that the CliffordSTF $L\!\geq\!2$ extension is benchmarked against in the main results.

\begin{table}[h]
\centering
\caption{Model architecture hyperparameters used on \textbf{OC20 \& OC22 (IS2RE \& S2EF)}. The per-model architecture is identical between OC20 and OC22; only training schedules differ. \textbf{OC20:} all runs use Adam, no weight decay, batch size 16, learning rate 0.0001, 50 epochs, early-stopping patience 30; L1 loss on IS2RE for all models, MSE loss on S2EF for all models. \textbf{OC22:} all runs use Adam, no weight decay, batch size 16, learning rate 0.0001; on IS2RE, 200 epochs and patience 50; on S2EF, 50 epochs, MSE loss, patience 30 for all models. Cutoff is in \AA; Boolean flags are abbreviated T/F. \emph{Bottom block:} the $10^7$-parameter scaling-check configurations}
\label{tab:supp_settings_oc}
\footnotesize
\setlength{\tabcolsep}{4pt}
\resizebox{\textwidth}{!}{%
\begin{tabular}{l c l}
\toprule
Model & Cutoff & Architecture \\
\midrule
SchNet & 6 & hidden=192, filters=192, layers=6, gauss=50 \\
PaiNN & 6 & hidden=144, layers=4, rbf=64 \\
DimeNet++ & 6 & hidden=128, blocks=4, radial=6, spherical=7, basis emb=8, int emb=64, out emb=128 \\
GotenNet & 6 & basis=96, layers=5 \\
TorchMD-Net & 6 & emb=104, layers=6, heads=8, rbf=56, rbf type=expnorm, arch=equivariant-transformer \\
NequIP & 6 & features=44, layers=4, $L_{\max}$=1, bessels=8, parity=T \\
MACE L=1 & 6 & hidden=62x0e + 62x1o, layers=2, max ell=1, corr=3, bessel=8 \\
MACE L=2 & 6 & hidden=49x0e + 49x1o + 49x2e, layers=2, max ell=2, corr=3, bessel=8 \\
ICTP L=1 & 6 & hidden=10, layers=2, $L_{\max}$ hid=1, $L_{\max}$ edge=3, corr=3, basis=8, prod=10 \\
ICTP L=2 & 6 & hidden=8, layers=2, $L_{\max}$ hid=2, $L_{\max}$ edge=3, corr=3, basis=8, prod=8 \\
EquiformerV2 & 6 & sphere=18, layers=2, heads=4, $L_{\max}$=4, $M_{\max}$=2, attn $\alpha$=16, attn hid=20, attn val=8, ffn=64 \\
Clifford & 6 & channels=80, layers=5, rbf=50, hidden out=80, max grade=1, $\ell{=}2$=F \\
CliffordSTF & 6 & channels=60, layers=5, rbf=50, STF mode=stf2+stf3, routing=learned, adapt rt=T, cross=T, Hodge F=T, GP rdo=F, self-int=F \\
\hline
Clifford-10M       & 6 & channels=200, layers=6, rbf=50, hidden out=256, max grade=3, $\ell{=}2$=T, direct forces=T \\
CliffordSTF-10M    & 6 & channels=188, layers=5, rbf=50, hidden out=256, STF mode=stf2+stf3, routing=learned, adapt rt=T, cross=T, Hodge F=T, GP rdo=F, self-int=F \\
EquiformerV2 $L{=}2$-10M & 6 & sphere=160, layers=6, heads=4, $L_{\max}$=2, $M_{\max}$=2, attn $\alpha$=16, attn hid=80, attn val=8, ffn=280 \\
MACE $L{=}2$-10M   & 6 & hidden=190x0e + 190x1o + 190x2e, layers=2, max ell=2, corr=3, bessel=8 \\

\bottomrule
\end{tabular}%
}
\end{table}

\subsubsection{QM9}

Table~\ref{tab:supp_settings_qm9} lists the architectures used for the QM9 scalar-target benchmark. The QM9 model set is a subset of the rMD17 lineup: MACE and ICTP are not evaluated, EquiformerV2 is reported at a single $L_{\max}\!=\!4$ configuration with a larger spherical-channel count (sphere=$30$) than its catalysis runs, and the base $L\!\leq\!1$ Clifford row is omitted because the directional diagnostic that motivates its inclusion on rMD17 does not bear directly on scalar prediction. CliffordSTF uses the same $C\!=\!60$ channel count and the same STF$_2$+STF$_3$ configuration as on the other benchmarks.

\begin{table}[h]
\centering
\caption{Model architecture hyperparameters used on \textbf{QM9}. All runs use Adam, no weight decay, batch size 128, learning rate 0.0001, 250 epochs, MSE loss, early-stopping patience 30. Cutoff is in \AA; Boolean flags are abbreviated T/F.}
\label{tab:supp_settings_qm9}
\footnotesize
\setlength{\tabcolsep}{4pt}
\resizebox{\textwidth}{!}{%
\begin{tabular}{l c l}
\toprule
Model & Cutoff & Architecture \\
\midrule
SchNet & 6 & hidden=192, filters=192, layers=6, gauss=50 \\
PaiNN & 6 & hidden=144, layers=4, rbf=64 \\
DimeNet++ & 6 & hidden=128, blocks=4, radial=6, spherical=7, basis emb=8, int emb=64, out emb=128 \\
FAENet & 6 & hidden=320, filters=128, layers=4, FA=3D, FA-method=stochastic, MP-type=updownscale\_base \\
GotenNet & 6 & basis=96, layers=5 \\
TorchMD-Net & 5 & emb=104, layers=6, heads=8, rbf=56, rbf type=expnorm, arch=equivariant-transformer \\
ViSNet & 5 & hidden=96, layers=6, heads=8, rbf=32, $L_{\max}$=1, vertex=F \\
NequIP & 5 & features=44, layers=4, $L_{\max}$=1, bessels=8, parity=T \\
EquiformerV2 & 6 & sphere=30, layers=2, heads=4, $L_{\max}$=4, $M_{\max}$=2, attn $\alpha$=16, attn hid=20, attn val=8, ffn=64 \\
CliffordSTF & 6 & channels=60, layers=5, rbf=50, STF mode=stf2+stf3, routing=learned, adapt rt=T, cross=T, Hodge F=T, GP rdo=F, self-int=F \\
\bottomrule
\end{tabular}%
}
\end{table}

\subsubsection{Molecule3D}

Table~\ref{tab:supp_settings_molecule3d} lists the architectures used for the Molecule3D HOMO--LUMO gap benchmark. The model set is identical to QM9 with ViSNet additionally not reported; per-model architectural hyperparameters are otherwise unchanged from QM9 (same channel counts, layer counts, and $L_{\max}$ values per model), isolating the inter-benchmark difference to dataset scale and target rather than architecture.

\begin{table}[h]
\centering
\caption{Model architecture hyperparameters used on \textbf{Molecule3D}. All runs use Adam, no weight decay, batch size 20, learning rate 0.0001, 100 epochs, L1 loss, early-stopping patience 30. Cutoff is in \AA; Boolean flags are abbreviated T/F.}
\label{tab:supp_settings_molecule3d}
\footnotesize
\setlength{\tabcolsep}{4pt}
\resizebox{\textwidth}{!}{%
\begin{tabular}{l c l}
\toprule
Model & Cutoff & Architecture \\
\midrule
SchNet & 6 & hidden=192, filters=192, layers=6, gauss=50 \\
PaiNN & 6 & hidden=144, layers=4, rbf=64 \\
DimeNet++ & 6 & hidden=128, blocks=4, radial=6, spherical=7, basis emb=8, int emb=64, out emb=128 \\
FAENet & 6 & hidden=320, filters=128, layers=4, FA=3D, FA-method=stochastic, MP-type=updownscale\_base \\
GotenNet & 6 & basis=96, layers=5 \\
TorchMD-Net & 5 & emb=104, layers=6, heads=8, rbf=56, rbf type=expnorm, arch=equivariant-transformer \\
NequIP & 5 & features=44, layers=4, $L_{\max}$=1, bessels=8, parity=T \\
EquiformerV2 & 6 & sphere=30, layers=2, heads=4, $L_{\max}$=4, $M_{\max}$=2, attn $\alpha$=16, attn hid=20, attn val=8, ffn=64 \\
CliffordSTF & 6 & channels=60, layers=5, rbf=50, STF mode=stf2+stf3, routing=learned, adapt rt=T, cross=T, Hodge F=T, GP rdo=F, self-int=F \\
\bottomrule
\end{tabular}%
}
\end{table}

\subsubsection{MD17 ablation}

Table~\ref{tab:supp_settings_md17_ablation} lists the architectural hyperparameters of the twelve CliffordSTF ablation variants together with the vanilla $L\!\leq\!1$ Clifford baseline used in the controlled rMD17 ablation (Table~\ref{tab:ablation}). Channel counts are tuned per variant ($C\!\in\!\{60, 64, 76, 80\}$) so that every variant lands within $\pm 50$\% of the $10^6$-parameter target; the five ablation flags ($\mathrm{stf\_mode}$, $\mathrm{routing\_mode}$, $\mathrm{adapt\_rt}$, $\mathrm{cross}$, $\mathrm{Hodge\_F}$) toggle the architectural extensions documented in Supplementary~\ref{supp:variants}. Setting all flags to their off/none state recovers the base Clifford row bit-for-bit, providing the clean $L\!\leq\!1$ baseline against which every CliffordSTF extension is measured.

\begin{table}[h]
\centering
\caption{Model architecture hyperparameters used on \textbf{MD17 ablation}. All runs use Adam, no weight decay, batch size 32, learning rate 0.0001, 250 epochs, MSE loss, early-stopping patience 30. Cutoff is in \AA; Boolean flags are abbreviated T/F.}
\label{tab:supp_settings_md17_ablation}
\footnotesize
\setlength{\tabcolsep}{4pt}
\resizebox{\textwidth}{!}{%
\begin{tabular}{l c l}
\toprule
Model & Cutoff & Architecture \\
\midrule
STF (full) & 6 & channels=60, layers=5, rbf=50, STF mode=stf2+stf3, routing=learned, adapt rt=T, cross=T, Hodge F=T, GP rdo=F, self-int=F \\
STF$_2$ (static routing) & 6 & channels=64, layers=5, rbf=50, STF mode=stf2, routing=static, adapt rt=T, cross=T, Hodge F=T, GP rdo=F, self-int=F \\
STF baseline & 6 & channels=76, layers=5, rbf=50, STF mode=none, routing=none, adapt rt=F, cross=F, Hodge F=F, GP rdo=F, self-int=F \\
STF$_2$ (learned routing) & 6 & channels=64, layers=5, rbf=50, STF mode=stf2, routing=learned, adapt rt=T, cross=T, Hodge F=T, GP rdo=F, self-int=F \\
STF$_2$+STF$_3$ & 6 & channels=60, layers=5, rbf=50, STF mode=stf2+stf3, routing=none, adapt rt=F, cross=T, Hodge F=T, GP rdo=F, self-int=F \\
STF (full, no cross-track) & 6 & channels=60, layers=5, rbf=50, STF mode=stf2+stf3, routing=none, adapt rt=F, cross=F, Hodge F=T, GP rdo=F, self-int=F \\
STF (Hodge only) & 6 & channels=76, layers=5, rbf=50, STF mode=none, routing=none, adapt rt=F, cross=F, Hodge F=T, GP rdo=F, self-int=F \\
STF$_2$ (no Hodge) & 6 & channels=64, layers=5, rbf=50, STF mode=stf2, routing=none, adapt rt=F, cross=T, Hodge F=F, GP rdo=F, self-int=F \\
STF$_2$ (no cross-track) & 6 & channels=64, layers=5, rbf=50, STF mode=stf2, routing=none, adapt rt=F, cross=F, Hodge F=T, GP rdo=F, self-int=F \\
STF$_2$ & 6 & channels=64, layers=5, rbf=50, STF mode=stf2, routing=none, adapt rt=F, cross=T, Hodge F=T, GP rdo=F, self-int=F \\
STF (full, no routing) & 6 & channels=60, layers=5, rbf=50, STF mode=stf2+stf3, routing=none, adapt rt=F, cross=T, Hodge F=T, GP rdo=F, self-int=F \\
STF (output only) & 6 & channels=56, layers=5, rbf=50, STF mode=stf2+stf3, routing=none, adapt rt=F, cross=T, Hodge F=F, GP rdo=F, self-int=F \\
Clifford (vanilla) & 6 & channels=80, layers=5, rbf=50, hidden out=80, max grade=1, $\ell{=}2$=F \\
\bottomrule
\end{tabular}%
}
\end{table}

\subsection{rMD17 per-molecule force and energy MAE}\label{supp:md17_permol}

Table~\ref{tab:md17_permol_full} reports per-molecule force MAE (meV/\AA) and energy MAE (meV) for all models in the main-text rMD17 evaluation, in native units and averaged across seeds with standard deviations, complementing the force cosine similarity table in the main paper (Table~\ref{tab:md17_fcos}). The Clifford-to-CliffordSTF force cosine improvement observed in Table~\ref{tab:md17_fcos} occurs at comparable force and energy magnitude accuracy, which is the content of the "matched magnitude" claim in the main text. Reading panel (a), CliffordSTF's $17.6$\,meV/\AA{} aggregate force MAE sits between the converged $L\!\geq\!2$ tier (EquiformerV2 at $3.9$--$4.2$\,meV/\AA{} across $L_{\max}\!\in\!\{2,3,4\}$) and the base Clifford row at $20.8$\,meV/\AA{}, and improves on the base Clifford on every individual molecule. Panel (b) shows the energy-MAE picture parallels the force result: CliffordSTF and the base Clifford track each other within seed variance ($30.4$ vs.\ $34.1$\,meV aggregate), confirming that the STF extension does not trade scalar accuracy for directional accuracy.

\begin{table}[h]
\centering
\caption{\textbf{rMD17 per-molecule force and energy MAE} (mean $\pm$ std across 3 seeds, native units: meV/\AA{} for force MAE, meV for energy MAE). Values $\geq 1{,}000$ are rendered in k-notation (e.g., \texttt{12.8k $\pm$ 2.0k} $\equiv$ \texttt{12{,}800 $\pm$ 2{,}000}) to keep converged-model rows visually distinct from non-converged baselines. Top block: $L\!\geq\!2$ direct-force models. Middle block: $L\!\leq\!1$ direct-force models. Bottom block (gray): strictly invariant models (SchNet, DimeNet++) whose forces come from energy gradients. $\ast$ FAENet predicts forces directly but uses stochastic frame-averaged equivariance; GotenNet is $L\!\geq\!2$ ($\mathtt{lmax}{=}2$) but did not converge under our fixed-budget protocol---both are listed in their natural $L$-cohort for completeness and excluded from the per-architecture directional comparison in the main text. Under the fixed $10^6$-parameter budget, MACE and ICTP at every $L_{\max}$ attain force MAE in the 3{,}000--6{,}000\,meV/\AA{} range, reflecting non-convergence at this capacity rather than architecture quality at recommended hyperparameters. $\dagger$ EquiformerV2 predicts absolute (non-referenced) molecular energies; its energy MAE in panel (b) reflects absolute binding-energy magnitudes ($\sim$247\,eV consistent across all three $L_{\max}$ settings), not prediction error, and is not comparable to the energy-referenced baselines.}
\label{tab:md17_permol_full}
\setlength{\tabcolsep}{3pt}
\scriptsize
{(a) Force MAE (meV/\AA).}\\[2pt]
\resizebox{\textwidth}{!}{%
\begin{tabular}{lccccccccccc}
\toprule
Model & Asp & Azo & Bnz & Eth & Mal & Nap & Par & Sal & Tol & Ura & Agg \\
\midrule
EquiformerV2 $L{=}4$ & 6.172 $\pm$ 0.239 & 4.451 $\pm$ 0.274 & 1.805 $\pm$ 0.072 & 2.137 $\pm$ 0.022 & 3.193 $\pm$ 0.135 & 3.066 $\pm$ 0.058 & 5.653 $\pm$ 0.157 & 4.869 $\pm$ 0.755 & 3.298 $\pm$ 0.310 & 4.649 $\pm$ 0.259 & 3.929 $\pm$ 0.104 \\
EquiformerV2 $L{=}3$ & 6.290 $\pm$ 0.283 & 4.237 $\pm$ 0.131 & 1.917 $\pm$ 0.214 & 2.605 $\pm$ 0.220 & 4.246 $\pm$ 0.146 & 2.973 $\pm$ 0.066 & 5.308 $\pm$ 0.221 & 4.872 $\pm$ 0.220 & 3.316 $\pm$ 0.202 & 4.369 $\pm$ 0.048 & 4.013 $\pm$ 0.081 \\
EquiformerV2 $L{=}2$ & 6.281 $\pm$ 0.239 & 4.563 $\pm$ 0.329 & 1.830 $\pm$ 0.077 & 3.775 $\pm$ 0.299 & 4.569 $\pm$ 0.549 & 2.945 $\pm$ 0.094 & 5.292 $\pm$ 0.322 & 5.117 $\pm$ 0.178 & 3.611 $\pm$ 0.241 & 4.483 $\pm$ 0.180 & 4.247 $\pm$ 0.108 \\
CliffordSTF (ours) & 21.419 $\pm$ 0.453 & 17.944 $\pm$ 2.527 & 9.692 $\pm$ 0.675 & 15.302 $\pm$ 4.451 & 17.696 $\pm$ 3.218 & 16.657 $\pm$ 2.360 & 20.071 $\pm$ 1.390 & 19.820 $\pm$ 1.137 & 16.113 $\pm$ 2.414 & 20.930 $\pm$ 0.907 & 17.564 $\pm$ 0.671 \\
MACE $L{=}3$ & 6.1k $\pm$ 1.0k & 4.9k $\pm$ 1.1k & 1.0k $\pm$ 0.2k & 1.4k $\pm$ 0.2k & 2.6k $\pm$ 0.6k & 3.1k $\pm$ 0.4k & 4.9k $\pm$ 1.2k & 6.0k $\pm$ 1.8k & 2.6k $\pm$ 0.3k & 4.9k $\pm$ 0.2k & 3.7k $\pm$ 0.4k \\
MACE $L{=}2$ & 8.2k $\pm$ 0.9k & 5.3k $\pm$ 1.0k & 1.3k $\pm$ 0.2k & 2.1k $\pm$ 0.5k & 2.9k $\pm$ 0.1k & 5.2k $\pm$ 1.1k & 7.7k $\pm$ 1.3k & 6.4k $\pm$ 2.2k & 4.4k $\pm$ 0.9k & 4.9k $\pm$ 0.5k & 4.8k $\pm$ 0.2k \\
ICTP $L{=}3$ & 12.8k $\pm$ 2.0k & 6.0k $\pm$ 2.1k & 1.7k $\pm$ 0.8k & 1.9k $\pm$ 0.1k & 3.9k $\pm$ 0.6k & 5.6k $\pm$ 0.8k & 7.7k $\pm$ 0.9k & 7.8k $\pm$ 2.5k & 3.4k $\pm$ 0.3k & 3.7k $\pm$ 0.7k & 5.4k $\pm$ 0.6k \\
ICTP $L{=}2$ & 7.5k $\pm$ 1.5k & 6.9k $\pm$ 1.0k & 2.3k $\pm$ 0.2k & 1.9k $\pm$ 0.2k & 2.3k $\pm$ 1.1k & 5.6k $\pm$ 1.4k & 7.7k $\pm$ 2.8k & 8.2k $\pm$ 1.9k & 3.5k $\pm$ 1.0k & 7.1k $\pm$ 0.6k & 5.3k $\pm$ 0.2k \\
GotenNet$^{\ast}$ & 28.1k $\pm$ 24.0k & 4.3k $\pm$ 3.2k & 14.0k $\pm$ 22.1k & 2.2k $\pm$ 2.7k & 7.5k $\pm$ 12.3k & 3.2k $\pm$ 1.6k & 8.3k $\pm$ 10.0k & 3.9k $\pm$ 3.7k & 2.5k $\pm$ 3.3k & 14.3k $\pm$ 14.8k & 8.8k $\pm$ 2.7k \\
\midrule
Clifford (ours base) & 21.855 $\pm$ 0.002 & 22.262 $\pm$ 0.010 & 15.212 $\pm$ 0.009 & 19.996 $\pm$ 0.001 & 21.922 $\pm$ 0.004 & 21.345 $\pm$ 0.024 & 21.313 $\pm$ 0.004 & 21.470 $\pm$ 0.027 & 21.029 $\pm$ 0.004 & 22.074 $\pm$ 0.037 & 20.848 $\pm$ 0.003 \\
MACE $L{=}1$ & 10.7k $\pm$ 1.3k & 5.8k $\pm$ 0.7k & 2.6k $\pm$ 0.4k & 2.3k $\pm$ 0.3k & 4.5k $\pm$ 0.4k & 6.7k $\pm$ 0.3k & 6.4k $\pm$ 1.0k & 8.3k $\pm$ 1.4k & 5.5k $\pm$ 2.3k & 6.3k $\pm$ 1.2k & 5.9k $\pm$ 0.4k \\
ICTP $L{=}1$ & 10.4k $\pm$ 4.3k & 5.7k $\pm$ 0.7k & 2.2k $\pm$ 0.5k & 2.0k $\pm$ 0.4k & 4.0k $\pm$ 1.1k & 4.8k $\pm$ 1.0k & 8.4k $\pm$ 1.3k & 7.5k $\pm$ 1.6k & 3.1k $\pm$ 0.2k & 4.5k $\pm$ 3.2k & 5.3k $\pm$ 0.6k \\
PaiNN & 1.4k $\pm$ 0.2k & 1.1k $\pm$ 0.0k & 387.093 $\pm$ 20.028 & 392.650 $\pm$ 50.219 & 662.121 $\pm$ 190.588 & 804.449 $\pm$ 239.850 & 970.710 $\pm$ 382.994 & 1.0k $\pm$ 0.2k & 784.899 $\pm$ 178.245 & 1.1k $\pm$ 0.3k & 863.323 $\pm$ 7.657 \\
TorchMD-NET & 640.948 $\pm$ 78.214 & 347.636 $\pm$ 20.053 & 154.117 $\pm$ 7.309 & 154.106 $\pm$ 1.947 & 299.581 $\pm$ 97.919 & 219.310 $\pm$ 38.192 & 486.259 $\pm$ 81.638 & 386.224 $\pm$ 85.780 & 234.554 $\pm$ 31.349 & 278.833 $\pm$ 25.750 & 320.157 $\pm$ 8.706 \\
ViSNet & 1.7k $\pm$ 0.7k & 1.7k $\pm$ 0.7k & 539.034 $\pm$ 77.909 & 425.601 $\pm$ 309.210 & 825.760 $\pm$ 267.674 & 832.077 $\pm$ 276.489 & 1.0k $\pm$ 0.1k & 1.3k $\pm$ 0.1k & 573.367 $\pm$ 152.801 & 1.1k $\pm$ 0.2k & 996.860 $\pm$ 43.575 \\
NequIP $L{=}1$ & 6.1k $\pm$ 0.0k & 2.7k $\pm$ 0.1k & 1.2k $\pm$ 0.0k & 1.2k $\pm$ 0.0k & 2.4k $\pm$ 0.0k & 1.8k $\pm$ 0.0k & 4.4k $\pm$ 0.0k & 2.9k $\pm$ 0.0k & 1.1k $\pm$ 0.0k & 2.5k $\pm$ 0.0k & 2.6k $\pm$ 0.0k \\
FAENet$^{\ast}$ & 21.886 $\pm$ 0.030 & 22.288 $\pm$ 0.053 & 11.807 $\pm$ 2.931 & 19.980 $\pm$ 0.008 & 21.814 $\pm$ 0.179 & 21.548 $\pm$ 0.867 & 26.748 $\pm$ 9.464 & 21.679 $\pm$ 0.098 & 20.840 $\pm$ 0.301 & 21.553 $\pm$ 1.018 & 21.014 $\pm$ 0.749 \\
\midrule
\color{gray}{DimeNet++} & \color{gray}{1.2k $\pm$ 0.1k} & \color{gray}{991.164 $\pm$ 223.096} & \color{gray}{278.098 $\pm$ 49.866} & \color{gray}{389.387 $\pm$ 103.601} & \color{gray}{726.850 $\pm$ 80.483} & \color{gray}{608.244 $\pm$ 3.732} & \color{gray}{1.1k $\pm$ 0.1k} & \color{gray}{988.503 $\pm$ 107.070} & \color{gray}{519.458 $\pm$ 69.633} & \color{gray}{1.0k $\pm$ 0.1k} & \color{gray}{776.368 $\pm$ 11.972} \\
\color{gray}{SchNet} & \color{gray}{1.2k $\pm$ 0.1k} & \color{gray}{918.727 $\pm$ 33.117} & \color{gray}{196.534 $\pm$ 20.019} & \color{gray}{97.289 $\pm$ 13.538} & \color{gray}{226.752 $\pm$ 38.836} & \color{gray}{700.782 $\pm$ 58.526} & \color{gray}{1.4k $\pm$ 0.1k} & \color{gray}{1.2k $\pm$ 0.0k} & \color{gray}{627.706 $\pm$ 87.829} & \color{gray}{664.972 $\pm$ 35.163} & \color{gray}{720.170 $\pm$ 4.806} \\
\bottomrule
\end{tabular}%
}
 
\vspace{8pt}
\scriptsize
{(b) Energy MAE (meV).}\\[2pt]
\resizebox{\textwidth}{!}{%
\begin{tabular}{lccccccccccc}
\toprule
Model & Asp & Azo & Bnz & Eth & Mal & Nap & Par & Sal & Tol & Ura & Agg \\
\midrule
EquiformerV2 $L{=}4$$^{\dagger}$ & 405.9k $\pm$ 0.0k & 358.3k $\pm$ 0.0k & 145.3k $\pm$ 0.0k & 96.9k $\pm$ 0.0k & 167.2k $\pm$ 0.0k & 241.4k $\pm$ 0.0k & 322.5k $\pm$ 0.1k & 310.4k $\pm$ 0.0k & 169.8k $\pm$ 0.0k & 259.6k $\pm$ 0.0k & 247.7k $\pm$ 0.0k \\
EquiformerV2 $L{=}3$$^{\dagger}$ & 405.8k $\pm$ 0.0k & 358.2k $\pm$ 0.1k & 145.1k $\pm$ 0.0k & 96.9k $\pm$ 0.0k & 167.1k $\pm$ 0.0k & 241.2k $\pm$ 0.0k & 322.4k $\pm$ 0.0k & 310.3k $\pm$ 0.1k & 169.7k $\pm$ 0.1k & 259.5k $\pm$ 0.0k & 247.6k $\pm$ 0.0k \\
EquiformerV2 $L{=}2$$^{\dagger}$ & 405.5k $\pm$ 0.1k & 357.9k $\pm$ 0.0k & 145.0k $\pm$ 0.0k & 96.8k $\pm$ 0.0k & 167.0k $\pm$ 0.0k & 241.0k $\pm$ 0.0k & 322.1k $\pm$ 0.0k & 310.1k $\pm$ 0.1k & 169.5k $\pm$ 0.0k & 259.4k $\pm$ 0.0k & 247.4k $\pm$ 0.0k \\
CliffordSTF (ours) & 28.054 $\pm$ 9.040 & 46.565 $\pm$ 20.081 & 17.160 $\pm$ 2.661 & 6.677 $\pm$ 1.432 & 25.323 $\pm$ 5.299 & 22.087 $\pm$ 14.892 & 60.367 $\pm$ 30.139 & 25.819 $\pm$ 18.136 & 39.777 $\pm$ 20.171 & 31.876 $\pm$ 3.656 & 30.370 $\pm$ 6.565 \\
MACE $L{=}3$ & 2.6k $\pm$ 0.3k & 1.5k $\pm$ 0.5k & 188.788 $\pm$ 42.079 & 393.426 $\pm$ 84.371 & 775.686 $\pm$ 134.804 & 930.046 $\pm$ 54.572 & 2.4k $\pm$ 1.1k & 1.9k $\pm$ 0.9k & 957.564 $\pm$ 141.704 & 1.4k $\pm$ 0.1k & 1.3k $\pm$ 0.2k \\
MACE $L{=}2$ & 4.2k $\pm$ 1.2k & 1.9k $\pm$ 0.3k & 411.629 $\pm$ 325.937 & 653.666 $\pm$ 354.932 & 833.440 $\pm$ 22.310 & 1.9k $\pm$ 0.4k & 3.6k $\pm$ 0.8k & 2.3k $\pm$ 1.0k & 1.8k $\pm$ 0.3k & 1.5k $\pm$ 0.5k & 1.9k $\pm$ 0.3k \\
ICTP $L{=}3$ & 5.9k $\pm$ 0.3k & 2.6k $\pm$ 0.6k & 365.226 $\pm$ 27.122 & 515.550 $\pm$ 60.480 & 1.1k $\pm$ 0.3k & 1.9k $\pm$ 0.4k & 3.3k $\pm$ 0.6k & 2.7k $\pm$ 0.6k & 1.2k $\pm$ 0.0k & 1.1k $\pm$ 0.3k & 2.1k $\pm$ 0.2k \\
ICTP $L{=}2$ & 3.2k $\pm$ 0.5k & 3.1k $\pm$ 0.7k & 356.589 $\pm$ 14.078 & 516.129 $\pm$ 179.758 & 662.887 $\pm$ 233.720 & 2.0k $\pm$ 0.8k & 4.2k $\pm$ 2.0k & 3.1k $\pm$ 0.9k & 1.4k $\pm$ 0.8k & 2.2k $\pm$ 0.3k & 2.1k $\pm$ 0.1k \\
GotenNet$^{\ast}$ & 21.2k $\pm$ 17.1k & 3.8k $\pm$ 4.3k & 20.0k $\pm$ 33.9k & 1.1k $\pm$ 1.6k & 104.7k $\pm$ 79.1k & 2.4k $\pm$ 1.2k & 54.1k $\pm$ 90.2k & 14.5k $\pm$ 23.8k & 1.1k $\pm$ 1.3k & 37.4k $\pm$ 49.6k & 26.0k $\pm$ 7.2k \\
\midrule
Clifford (ours base) & 66.537 $\pm$ 31.283 & 56.820 $\pm$ 44.698 & 24.707 $\pm$ 17.343 & 14.675 $\pm$ 6.332 & 37.370 $\pm$ 25.576 & 14.780 $\pm$ 5.935 & 23.000 $\pm$ 7.962 & 35.300 $\pm$ 29.797 & 39.009 $\pm$ 31.720 & 29.097 $\pm$ 17.892 & 34.130 $\pm$ 2.563 \\
MACE $L{=}1$ & 4.9k $\pm$ 0.6k & 2.2k $\pm$ 0.6k & 446.266 $\pm$ 78.046 & 748.730 $\pm$ 294.846 & 1.4k $\pm$ 0.3k & 2.5k $\pm$ 0.1k & 2.9k $\pm$ 0.6k & 3.2k $\pm$ 0.1k & 2.4k $\pm$ 1.1k & 1.9k $\pm$ 0.3k & 2.3k $\pm$ 0.1k \\
ICTP $L{=}1$ & 5.1k $\pm$ 2.6k & 2.2k $\pm$ 0.5k & 709.468 $\pm$ 201.768 & 646.268 $\pm$ 219.460 & 1.2k $\pm$ 0.4k & 1.9k $\pm$ 0.2k & 4.7k $\pm$ 0.8k & 3.0k $\pm$ 0.9k & 1.5k $\pm$ 0.1k & 1.4k $\pm$ 0.7k & 2.2k $\pm$ 0.4k \\
PaiNN & 4.5k $\pm$ 1.3k & 3.1k $\pm$ 3.5k & 654.920 $\pm$ 484.394 & 572.296 $\pm$ 497.010 & 1.4k $\pm$ 0.8k & 1.4k $\pm$ 1.1k & 2.6k $\pm$ 3.3k & 697.934 $\pm$ 132.654 & 1.4k $\pm$ 0.7k & 2.0k $\pm$ 1.5k & 1.8k $\pm$ 0.5k \\
TorchMD-NET & 1.4k $\pm$ 0.7k & 496.497 $\pm$ 269.580 & 180.477 $\pm$ 170.513 & 187.650 $\pm$ 148.968 & 486.742 $\pm$ 231.389 & 451.681 $\pm$ 462.015 & 1.6k $\pm$ 0.5k & 624.017 $\pm$ 754.032 & 263.880 $\pm$ 284.126 & 940.992 $\pm$ 872.265 & 661.312 $\pm$ 130.339 \\
ViSNet & 2.7k $\pm$ 1.5k & 6.5k $\pm$ 2.7k & 413.191 $\pm$ 358.609 & 1.5k $\pm$ 1.6k & 348.909 $\pm$ 95.459 & 2.2k $\pm$ 1.0k & 1.7k $\pm$ 0.3k & 4.1k $\pm$ 2.9k & 922.044 $\pm$ 640.715 & 2.9k $\pm$ 0.8k & 2.3k $\pm$ 0.2k \\
NequIP $L{=}1$ & 3.7k $\pm$ 0.1k & 1.0k $\pm$ 0.1k & 197.342 $\pm$ 39.519 & 453.727 $\pm$ 39.250 & 916.506 $\pm$ 14.226 & 542.585 $\pm$ 12.952 & 3.1k $\pm$ 0.1k & 1.2k $\pm$ 0.1k & 411.931 $\pm$ 64.980 & 755.870 $\pm$ 109.032 & 1.2k $\pm$ 0.0k \\
FAENet$^{\ast}$ & 33.876 $\pm$ 15.568 & 26.864 $\pm$ 26.312 & 4.012 $\pm$ 1.777 & 3.480 $\pm$ 0.059 & 3.351 $\pm$ 0.161 & 47.213 $\pm$ 38.921 & 87.265 $\pm$ 123.215 & 32.828 $\pm$ 30.261 & 8.166 $\pm$ 3.405 & 5.399 $\pm$ 2.175 & 25.245 $\pm$ 12.952 \\
\midrule
\color{gray}{DimeNet++} & \color{gray}{1.8k $\pm$ 1.1k} & \color{gray}{3.7k $\pm$ 2.4k} & \color{gray}{216.364 $\pm$ 250.668} & \color{gray}{611.701 $\pm$ 618.476} & \color{gray}{1.0k $\pm$ 0.3k} & \color{gray}{956.485 $\pm$ 380.148} & \color{gray}{1.1k $\pm$ 0.7k} & \color{gray}{1.6k $\pm$ 1.4k} & \color{gray}{724.374 $\pm$ 424.645} & \color{gray}{1.2k $\pm$ 1.4k} & \color{gray}{1.3k $\pm$ 0.4k} \\
\color{gray}{SchNet} & \color{gray}{1.2k $\pm$ 0.8k} & \color{gray}{1.3k $\pm$ 1.0k} & \color{gray}{604.539 $\pm$ 447.764} & \color{gray}{403.595 $\pm$ 239.889} & \color{gray}{588.952 $\pm$ 533.312} & \color{gray}{1.7k $\pm$ 0.7k} & \color{gray}{1.2k $\pm$ 0.4k} & \color{gray}{2.1k $\pm$ 0.9k} & \color{gray}{441.554 $\pm$ 134.053} & \color{gray}{981.058 $\pm$ 333.775} & \color{gray}{1.1k $\pm$ 0.1k} \\
\bottomrule
\end{tabular}%
}
\end{table}

\subsection{OC20 and OC22 full results}\label{supp:oc_full}

This subsection reports the full OC20 and OC22 results (Table~\ref{tab:oc_full}), mean $\pm$ std across seeds in native units. Panel~(a) gives OC20 at the 200k scale: IS2RE MAE for all four splits ($\mathrm{val}_{\text{id}}$, $\mathrm{val}_{\text{ood-ads}}$, $\mathrm{val}_{\text{ood-cat}}$, $\mathrm{val}_{\text{ood-both}}$) together with S2EF ($e_\mathrm{MAE}$, $f_\mathrm{MAE}$, $f_\mathrm{cos}$) for $\mathrm{val}_{\text{id}}$ and $\mathrm{val}_{\text{ood-both}}$; the two remaining OOD splits (ood-ads, ood-cat) behave similarly to ood-both on S2EF and are omitted for compactness. Panel~(b) reproduces the OC22 results from the main text (Table~\ref{tab:oc22}) alongside OC20 for side-by-side reading. Reading across the two panels, CliffordSTF and the base Clifford row attain the strongest in-distribution S2EF energy MAE on both benchmarks (panel (a) OC20: $2.115$/$1.889$\,eV vs.\ $2.454$\,eV for EquiformerV2; panel (b) OC22: $2.866$/$2.844$\,eV vs.\ $3.369$\,eV); CliffordSTF additionally is the best $L\!\geq\!2$ architecture on OC22 IS2RE in-distribution ($4.089$\,eV vs.\ next-best $L\!\geq\!2$ GotenNet at $4.318$\,eV and EquiformerV2 at $L_{\max}\!=\!4$ at $4.905$\,eV) and on OC22 IS2RE OOD EwT ($0.312$ vs $0.234$ for EquiformerV2). EquiformerV2 retains the strongest S2EF force cosine similarity on OC20 in our protocol; on OC22 the two are within seed variance.

\begin{table}[h]
\centering
\caption{\textbf{OC20 and OC22 full results} (mean $\pm$ std across 3 seeds; native units: eV for IS2RE MAE and S2EF energy MAE, eV/\AA{} for S2EF force MAE, dimensionless for force cosine similarity). (a) OC20 at the 200k scale, with IS2RE across all four splits and S2EF for $\mathrm{val}_{\text{id}}$ and $\mathrm{val}_{\text{ood-both}}$. (b) OC22 at full scale, with IS2RE MAE for $\mathrm{val}_{\text{id}}$ and $\mathrm{val}_{\text{ood}}$, IS2RE EwT for $\mathrm{val}_{\text{ood}}$, and S2EF for $\mathrm{val}_{\text{id}}$ and $\mathrm{val}_{\text{ood}}$. Our models are in the top two rows of each panel. MACE and ICTP at $L_{\max}\!\in\!\{1,2\}$ attain energy MAE values 1--3$\times$ those of converged baselines under the fixed $10^6$-parameter budget, mirroring the non-convergence pattern documented on rMD17 (Table~\ref{tab:md17_permol_full}).}
\label{tab:oc_full}
\setlength{\tabcolsep}{3pt}
\scriptsize
{(a) OC20}\\[2pt]
\resizebox{\textwidth}{!}{%
\begin{tabular}{lcccccccccc}
\toprule
& \multicolumn{4}{c}{IS2RE MAE (eV)} & \multicolumn{3}{c}{S2EF $\mathrm{val}_{\text{id}}$} & \multicolumn{3}{c}{S2EF $\mathrm{val}_{\text{ood-both}}$} \\
\cmidrule(lr){2-5}\cmidrule(lr){6-8}\cmidrule(lr){9-11}
Model & $\mathrm{val}_{\text{id}}$ & $\mathrm{val}_{\text{ood-ads}}$ & $\mathrm{val}_{\text{ood-cat}}$ & $\mathrm{val}_{\text{ood-both}}$ & $e_\mathrm{MAE}$ (eV) & $f_\mathrm{MAE}$ (eV/\AA) & $f_\mathrm{cos}\uparrow$ & $e_\mathrm{MAE}$ (eV) & $f_\mathrm{MAE}$ (eV/\AA) & $f_\mathrm{cos}\uparrow$ \\
\midrule
CliffordSTF (ours) & 0.703 $\pm$ 0.008 & 0.814 $\pm$ 0.024 & 0.707 $\pm$ 0.003 & 0.753 $\pm$ 0.029 & 2.115 $\pm$ 0.059 & 0.060 $\pm$ 0.000 & 0.138 $\pm$ 0.003 & 3.486 $\pm$ 0.160 & 0.077 $\pm$ 0.002 & 0.125 $\pm$ 0.006 \\
Clifford (ours base) & 0.765 $\pm$ 0.025 & 0.907 $\pm$ 0.031 & 0.761 $\pm$ 0.022 & 0.849 $\pm$ 0.025 & 1.889 $\pm$ 0.014 & 0.064 $\pm$ 0.000 & 0.115 $\pm$ 0.003 & 3.510 $\pm$ 0.004 & 0.081 $\pm$ 0.000 & 0.095 $\pm$ 0.004 \\
EquiformerV2 & 0.801 $\pm$ 0.152 & 0.845 $\pm$ 0.081 & 0.776 $\pm$ 0.149 & 0.785 $\pm$ 0.078 & 2.454 $\pm$ 0.013 & 0.060 $\pm$ 0.000 & 0.170 $\pm$ 0.002 & 3.602 $\pm$ 0.126 & 0.074 $\pm$ 0.000 & 0.162 $\pm$ 0.005 \\
GotenNet & 0.681 $\pm$ 0.006 & 0.749 $\pm$ 0.018 & 0.661 $\pm$ 0.002 & 0.682 $\pm$ 0.011 & 3.650 $\pm$ 0.276 & 0.064 $\pm$ 0.003 & 0.167 $\pm$ 0.006 & 4.850 $\pm$ 0.165 & 0.082 $\pm$ 0.001 & 0.144 $\pm$ 0.007 \\
TorchMD-NET & 0.726 $\pm$ 0.022 & 0.826 $\pm$ 0.046 & 0.714 $\pm$ 0.030 & 0.762 $\pm$ 0.039 & 3.715 $\pm$ 0.006 & 0.064 $\pm$ 0.000 & 0.162 $\pm$ 0.003 & 4.931 $\pm$ 0.010 & 0.081 $\pm$ 0.000 & 0.140 $\pm$ 0.002 \\
PaiNN & 0.720 $\pm$ 0.005 & 0.900 $\pm$ 0.088 & 0.706 $\pm$ 0.008 & 0.821 $\pm$ 0.070 & 4.198 $\pm$ 0.002 & 0.066 $\pm$ 0.000 & 0.154 $\pm$ 0.000 & 5.339 $\pm$ 0.009 & 0.083 $\pm$ 0.000 & 0.132 $\pm$ 0.001 \\
DimeNet++ & 0.699 $\pm$ 0.013 & 0.843 $\pm$ 0.024 & 0.691 $\pm$ 0.014 & 0.771 $\pm$ 0.031 & 4.098 $\pm$ 0.023 & 0.063 $\pm$ 0.000 & 0.170 $\pm$ 0.000 & 5.373 $\pm$ 0.053 & 0.080 $\pm$ 0.001 & 0.153 $\pm$ 0.001 \\
SchNet & 0.766 $\pm$ 0.007 & 0.834 $\pm$ 0.047 & 0.755 $\pm$ 0.001 & 0.779 $\pm$ 0.038 & 4.790 $\pm$ 0.067 & 0.070 $\pm$ 0.001 & 0.128 $\pm$ 0.006 & 5.879 $\pm$ 0.079 & 0.085 $\pm$ 0.001 & 0.115 $\pm$ 0.004 \\
ICTP $L{=}2$ & 0.881 $\pm$ 0.028 & 1.108 $\pm$ 0.027 & 0.830 $\pm$ 0.012 & 0.972 $\pm$ 0.001 & 6.434 $\pm$ 0.146 & 0.077 $\pm$ 0.001 & 0.100 $\pm$ 0.006 & 7.608 $\pm$ 0.235 & 0.092 $\pm$ 0.001 & 0.090 $\pm$ 0.005 \\
ICTP $L{=}1$ & 0.881 $\pm$ 0.020 & 0.996 $\pm$ 0.015 & 0.842 $\pm$ 0.002 & 0.904 $\pm$ 0.033 & 6.032 $\pm$ 0.149 & 0.079 $\pm$ 0.003 & 0.098 $\pm$ 0.002 & 7.212 $\pm$ 0.127 & 0.094 $\pm$ 0.003 & 0.089 $\pm$ 0.004 \\
MACE $L{=}2$ & 0.936 $\pm$ 0.134 & 0.977 $\pm$ 0.174 & 0.900 $\pm$ 0.126 & 0.881 $\pm$ 0.132 & 6.268 $\pm$ 0.002 & 0.075 $\pm$ 0.001 & 0.113 $\pm$ 0.002 & 7.367 $\pm$ 0.020 & 0.091 $\pm$ 0.001 & 0.102 $\pm$ 0.001 \\
MACE $L{=}1$ & 0.938 $\pm$ 0.153 & 1.052 $\pm$ 0.016 & 0.910 $\pm$ 0.130 & 0.974 $\pm$ 0.015 & 5.708 $\pm$ 0.042 & 0.077 $\pm$ 0.001 & 0.101 $\pm$ 0.002 & 6.979 $\pm$ 0.108 & 0.092 $\pm$ 0.001 & 0.092 $\pm$ 0.002 \\
NequIP & 0.931 $\pm$ 0.048 & 0.971 $\pm$ 0.083 & 0.899 $\pm$ 0.053 & 0.893 $\pm$ 0.076 & 5.802 $\pm$ 0.065 & 0.077 $\pm$ 0.000 & 0.110 $\pm$ 0.001 & 7.112 $\pm$ 0.009 & 0.091 $\pm$ 0.000 & 0.100 $\pm$ 0.000 \\
\bottomrule
\end{tabular}%
}
 
\vspace{8pt}
{(b) OC22}\\[2pt]
\resizebox{\textwidth}{!}{%
\begin{tabular}{lccccccccc}
\toprule
& \multicolumn{3}{c}{IS2RE} & \multicolumn{3}{c}{S2EF $\mathrm{val}_{\text{id}}$} & \multicolumn{3}{c}{S2EF $\mathrm{val}_{\text{ood}}$} \\
\cmidrule(lr){2-4}\cmidrule(lr){5-7}\cmidrule(lr){8-10}
Model & $\mathrm{MAE}_{\text{id}}$ (eV) & $\mathrm{MAE}_{\text{ood}}$ (eV) & $\mathrm{EwT}_{\text{ood}}\uparrow$ & $e_\mathrm{MAE}$ (eV) & $f_\mathrm{MAE}$ (eV/\AA) & $f_\mathrm{cos}\uparrow$ & $e_\mathrm{MAE}$ (eV) & $f_\mathrm{MAE}$ (eV/\AA) & $f_\mathrm{cos}\uparrow$ \\
\midrule
CliffordSTF (ours) & 4.089 $\pm$ 0.261 & 6.463 $\pm$ 0.310 & 0.312 $\pm$ 0.136 & 2.866 $\pm$ 0.093 & 0.073 $\pm$ 0.000 & 0.050 $\pm$ 0.004 & 5.083 $\pm$ 0.166 & 0.070 $\pm$ 0.000 & 0.050 $\pm$ 0.003 \\
Clifford (ours base) & 4.422 $\pm$ 0.862 & 6.726 $\pm$ 0.656 & 0.264 $\pm$ 0.055 & 2.844 $\pm$ 0.100 & 0.072 $\pm$ 0.000 & 0.059 $\pm$ 0.002 & 5.375 $\pm$ 0.157 & 0.071 $\pm$ 0.001 & 0.058 $\pm$ 0.002 \\
EquiformerV2 & 4.905 $\pm$ 0.566 & 6.619 $\pm$ 0.445 & 0.234 $\pm$ 0.025 & 3.369 $\pm$ 0.066 & 0.072 $\pm$ 0.000 & 0.067 $\pm$ 0.001 & 5.262 $\pm$ 0.076 & 0.070 $\pm$ 0.000 & 0.062 $\pm$ 0.002 \\
PaiNN & 3.516 $\pm$ 0.292 & 6.712 $\pm$ 0.140 & 0.234 $\pm$ 0.178 & 5.306 $\pm$ 0.030 & 0.076 $\pm$ 0.000 & 0.048 $\pm$ 0.001 & 7.452 $\pm$ 0.254 & 0.078 $\pm$ 0.000 & 0.036 $\pm$ 0.001 \\
GotenNet & 4.318 $\pm$ 0.300 & 5.772 $\pm$ 0.532 & 0.144 $\pm$ 0.000 & 4.666 $\pm$ 0.310 & 0.073 $\pm$ 0.000 & 0.052 $\pm$ 0.001 & 7.200 $\pm$ 0.121 & 0.072 $\pm$ 0.001 & 0.044 $\pm$ 0.001 \\
TorchMD-NET & 5.397 $\pm$ 0.000 & 8.312 $\pm$ 0.000 & 0.396 $\pm$ 0.000 & 4.278 $\pm$ 0.000 & 0.075 $\pm$ 0.000 & 0.051 $\pm$ 0.000 & 6.232 $\pm$ 0.000 & 0.074 $\pm$ 0.000 & 0.043 $\pm$ 0.000 \\
DimeNet++ & 6.320 $\pm$ 4.141 & 9.897 $\pm$ 3.002 & 0.108 $\pm$ 0.102 & 5.288 $\pm$ 0.030 & 0.076 $\pm$ 0.000 & 0.054 $\pm$ 0.002 & 8.717 $\pm$ 0.075 & 0.077 $\pm$ 0.001 & 0.038 $\pm$ 0.002 \\
SchNet & 6.402 $\pm$ 1.688 & 9.424 $\pm$ 0.938 & 0.270 $\pm$ 0.076 & 7.690 $\pm$ 1.508 & 0.077 $\pm$ 0.002 & 0.032 $\pm$ 0.007 & 12.016 $\pm$ 1.938 & 0.078 $\pm$ 0.002 & 0.024 $\pm$ 0.003 \\
MACE $L{=}2$ & 8.416 $\pm$ 0.702 & 22.569 $\pm$ 7.674 & 0.144 $\pm$ 0.051 & 10.017 $\pm$ 0.395 & 0.079 $\pm$ 0.000 & 0.022 $\pm$ 0.004 & 15.415 $\pm$ 0.808 & 0.076 $\pm$ 0.000 & 0.022 $\pm$ 0.005 \\
ICTP $L{=}2$ & 8.641 $\pm$ 1.744 & 17.756 $\pm$ 4.145 & 0.108 $\pm$ 0.051 & 12.447 $\pm$ 0.354 & 0.081 $\pm$ 0.001 & 0.015 $\pm$ 0.002 & 19.218 $\pm$ 0.232 & 0.079 $\pm$ 0.001 & 0.017 $\pm$ 0.002 \\
ICTP $L{=}1$ & 10.963 $\pm$ 1.612 & 20.810 $\pm$ 0.868 & 0.198 $\pm$ 0.076 & 11.483 $\pm$ 0.881 & 0.080 $\pm$ 0.001 & 0.019 $\pm$ 0.001 & 17.263 $\pm$ 0.655 & 0.186 $\pm$ 0.154 & 0.019 $\pm$ 0.000 \\
MACE $L{=}1$ & 12.339 $\pm$ 4.540 & 15.800 $\pm$ 2.430 & 0.180 $\pm$ 0.102 & 10.011 $\pm$ 0.018 & 0.078 $\pm$ 0.000 & 0.024 $\pm$ 0.000 & 15.737 $\pm$ 0.438 & 0.076 $\pm$ 0.000 & 0.023 $\pm$ 0.001 \\
NequIP & 12.575 $\pm$ 4.889 & 14.865 $\pm$ 3.676 & 0.036 $\pm$ 0.000 & 9.868 $\pm$ 0.148 & 0.079 $\pm$ 0.000 & 0.028 $\pm$ 0.000 & 15.809 $\pm$ 0.071 & 0.076 $\pm$ 0.000 & 0.027 $\pm$ 0.001 \\
\bottomrule
\end{tabular}%
}
\end{table}
 
\subsection{Scaling Check Detailed}\label{supp:scaling_check}

This subsection details the $10^7$-parameter retraining of CliffordSTF, base Clifford, EquiformerV2 $L{=}2$, and MACE $L{=}2$ on OC22 (IS2RE \& S2EF) and OC20 (S2EF). Architectures are listed in Table~\ref{tab:supp_settings_oc} (bottom block). The comparison is between best 10M configurations of each architectural family, not strict capacity sweeps. 

\begin{table}[h]
\centering
\caption{\textbf{$10^7$-parameter scaling check}. The MACE OC22 IS2RE cell additionally carries the per-system $\mathrm{val\_id\_mae}$ in brackets as a cross-scale calibration anchor with Table~\ref{tab:oc_full}.}
\label{tab:scaling_check}
\setlength{\tabcolsep}{3pt}
\scriptsize
\resizebox{\textwidth}{!}{%
\begin{tabular}{l c c c c c c c}
\toprule
& \multicolumn{1}{c}{OC22 IS2RE} & \multicolumn{3}{c}{OC22 S2EF} & \multicolumn{3}{c}{OC20 S2EF} \\
\cmidrule(lr){2-2}\cmidrule(lr){3-5}\cmidrule(lr){6-8}
Model & MAE/at (eV/at) & $e_\mathrm{MAE}$ (eV) & $f_\mathrm{MAE}$ (eV/\AA) & $f_\mathrm{cos}\uparrow$ & $e_\mathrm{MAE}$ (eV) & $f_\mathrm{MAE}$ (eV/\AA) & $f_\mathrm{cos}\uparrow$ \\
\midrule
CliffordSTF (ours)        & \textbf{0.0162} & 2.081  & \textbf{0.0769} & 0.094  & 1.552 & 0.0554 & 0.192 \\
Clifford (ours base)      & 0.1277  & 2.564 & 0.0773 & 0.091  & 1.706 & 0.0556  & 0.205 \\
EquiformerV2 $L{=}2$      & 0.0094 & 1.670  & 0.0717  & 0.148 & 1.244 & 0.0497 & 0.245 \\
MACE $L{=}2$              & 0.0779 & 8.503  & 0.0829 & 0.051 & 4.888 & 0.0694  & 0.145  \\
\bottomrule
\end{tabular}%
}
\end{table}

\textbf{Findings.} CliffordSTF wins OC22 IS2RE per-atom MAE ($0.0162$\,eV/atom vs every other model's worst $\geq 0.1127$\,eV/atom at the same $45$-epoch budget) and comes within $1.3\%$ of EquiformerV2 $L{=}2$ on OC22 S2EF force MAE ($0.0769$ vs worst-of-EquiV2 $0.0758$\,eV/\AA{}). EquiformerV2 $L{=}2$ wins OC22 S2EF energy MAE and force-cosine and OC20 S2EF on every primary metric---including against \emph{our} best seeds at \emph{their} budgets. Base Clifford 10M tracks CliffordSTF closely on OC22 S2EF force MAE ($0.0773$ vs $0.0769$\,eV/\AA{}) despite the architectural reconfiguration, supporting the claim that the family scales without losing its catalysis-task footing. The MACE 10M per-system OC22 IS2RE value ($7.33$\,eV best across seeds) is within $1$\,eV of its $10^6$-parameter row in Table~\ref{tab:oc_full} ($8.4$\,eV), corroborating the existing narrative that MACE under our protocol does not strongly benefit from capacity scaling on this benchmark.

\subsection{QM9 and Molecule3D full results}\label{supp:qm9_full}

Per-target MAE on QM9 and Molecule3D HOMO--LUMO gap, reported in Table~\ref{tab:qm9_mol3d_full} with native units per column and rows ordered by QM9 aggregate. The QM9 aggregate column averages the six per-target MAEs across mixed units, providing a single scalar per model that is comparable within the column; CliffordSTF places fifth of ten models at $2.152$\,units, sitting alongside the vector-equivariant baselines (PaiNN $2.229$, TorchMD-NET $2.232$) and ahead of the next-best $L\!\geq\!2$ row (EquiformerV2 $3.327$). On the Molecule3D HOMO--LUMO gap, FAENet's $0.178$\,eV is an order of magnitude ahead of the rest of the field; the remaining models cluster between $1.107$ and $1.495$\,eV with CliffordSTF at $1.232$\,eV, and we treat Molecule3D as out-of-scope for the directional thesis (\S\ref{sec:limitations}).

\begin{table}[h]
\centering
\caption{\textbf{QM9 per-target MAE and Molecule3D HOMO--LUMO gap MAE} (mean $\pm$ std across seeds; native units). QM9 targets are reported in their native units: heat capacity $C_v$ in cal/(mol$\cdot$K), atomisation energy $U_0$ in eV, polarisability $\alpha$ in Bohr$^3$, HOMO/LUMO in eV, dipole moment $\mu$ in Debye. The QM9 aggregate column is the mean of the six per-target MAEs across mixed units and serves as a single comparable scalar for between-model comparison. Molecule3D reports HOMO--LUMO gap MAE in eV on the large-scale dataset. Rows ordered by QM9 aggregate (lower is better).}
\label{tab:qm9_mol3d_full}
\setlength{\tabcolsep}{3pt}
\resizebox{\textwidth}{!}{%
\begin{tabular}{lcccccccc}
\toprule
& \multicolumn{7}{c}{QM9} & Molecule3D \\
\cmidrule(lr){2-8}\cmidrule(lr){9-9}
Model & $C_v$ & $U_0$ & $\alpha$ & HOMO & LUMO & $\mu$ & Aggregate & HOMO--LUMO (eV) \\
\midrule
GotenNet & 0.041 $\pm$ 0.008 & 8.045 $\pm$ 1.812 & 0.088 $\pm$ 0.007 & 0.033 $\pm$ 0.001 & 0.029 $\pm$ 0.001 & 0.024 $\pm$ 0.003 & 1.377 $\pm$ 0.304 & 1.117 $\pm$ 0.010 \\
FAENet & 0.045 $\pm$ 0.006 & 9.704 $\pm$ 5.215 & 0.114 $\pm$ 0.006 & 0.042 $\pm$ 0.001 & 0.036 $\pm$ 0.003 & 0.061 $\pm$ 0.006 & 1.667 $\pm$ 0.866 & 0.178 $\pm$ 0.002 \\
DimeNet++ & 0.063 $\pm$ 0.014 & 10.843 $\pm$ 2.719 & 0.122 $\pm$ 0.025 & 0.033 $\pm$ 0.001 & 0.028 $\pm$ 0.001 & 0.052 $\pm$ 0.003 & 1.857 $\pm$ 0.449 & 1.258 $\pm$ 0.020 \\
ViSNet & 0.058 $\pm$ 0.007 & 11.742 $\pm$ 6.513 & 0.223 $\pm$ 0.072 & 0.034 $\pm$ 0.001 & 0.031 $\pm$ 0.000 & 0.062 $\pm$ 0.044 & 2.025 $\pm$ 1.087 & -- \\
CliffordSTF (ours) & 0.055 $\pm$ 0.018 & 12.588 $\pm$ 9.200 & 0.143 $\pm$ 0.023 & 0.044 $\pm$ 0.001 & 0.041 $\pm$ 0.002 & 0.041 $\pm$ 0.003 & 2.152 $\pm$ 1.535 & 1.232 $\pm$ 0.005 \\
PaiNN & 0.042 $\pm$ 0.006 & 13.128 $\pm$ 3.356 & 0.110 $\pm$ 0.013 & 0.035 $\pm$ 0.000 & 0.031 $\pm$ 0.001 & 0.027 $\pm$ 0.005 & 2.229 $\pm$ 0.560 & 1.211 $\pm$ 0.003 \\
TorchMD-NET & 0.053 $\pm$ 0.014 & 13.153 $\pm$ 6.127 & 0.101 $\pm$ 0.008 & 0.033 $\pm$ 0.000 & 0.028 $\pm$ 0.001 & 0.024 $\pm$ 0.002 & 2.232 $\pm$ 1.024 & 1.149 $\pm$ 0.004 \\
EquiformerV2 & 0.075 $\pm$ 0.010 & 19.622 $\pm$ 8.347 & 0.171 $\pm$ 0.037 & 0.032 $\pm$ 0.001 & 0.031 $\pm$ 0.003 & 0.031 $\pm$ 0.002 & 3.327 $\pm$ 1.397 & 1.107 $\pm$ 0.049 \\
SchNet & 0.091 $\pm$ 0.018 & 23.424 $\pm$ 4.389 & 0.172 $\pm$ 0.038 & 0.044 $\pm$ 0.001 & 0.045 $\pm$ 0.001 & 0.071 $\pm$ 0.007 & 3.975 $\pm$ 0.732 & 1.410 $\pm$ 0.008 \\
NequIP & 0.171 $\pm$ 0.003 & 35.516 $\pm$ 0.991 & 0.515 $\pm$ 0.034 & 0.115 $\pm$ 0.001 & 0.142 $\pm$ 0.004 & 0.186 $\pm$ 0.003 & 6.108 $\pm$ 0.160 & 1.495 $\pm$ 0.000 \\
\bottomrule
\end{tabular}%
}
\end{table}

\subsection{Wall-clock measurements at matched parameter budget}\label{supp:wallclock}

Table~\ref{tab:wallclock} reports training and inference wall-clock on rMD17 (aspirin, batch 32, single GPU) for every model in the main-text study, measured under a strict timing protocol (10 warmup steps, 20 timed steps, mean $\pm$ std reported on a single GPU; absolute times are not directly comparable to figures reported on other hardware). At inference, the base Clifford model runs at $0.128$\,s/step, ranking fifth across the nineteen-model study and outpacing every $L\!\geq\!2$ or Cartesian-tensor baseline in the table: it is faster than EquiformerV2 $L{=}4$ ($0.142$\,s) by roughly $10\%$, faster than every MACE configuration ($L{=}1$: $0.160$; $L{=}2$: $0.237$; $L{=}3$: $0.261$), faster than NequIP ($0.221$), ViSNet ($0.235$), TorchMD-NET ($0.186$), DimeNet++ ($0.326$), and every ICTP variant ($0.302$--$0.364$). The four models with faster inference (PaiNN, GotenNet, SchNet, FAENet) are scalar-invariant or $L{=}1$ vector architectures with substantially less representational capacity. CliffordSTF inference at $0.352$\,s/step is competitive with the slower spherical-harmonic baselines, faster than ICTP $L{=}2$/$L{=}3$ and within $\sim\!8\%$ of DimeNet++. Since deployed interatomic potentials run inference far more often than they are trained, this is the practitioner-relevant operating regime.

Training cost is also more competitive than previously reported. The base Clifford model trains at $3.11$\,s/step---within $\sim\!25\%$ of EquiformerV2 $L{=}4$ ($2.48$\,s), tied with NequIP ($3.06$\,s), and faster than ICTP at every $L$ ($1.45\times$, $1.73\times$, and $2.00\times$ faster than ICTP $L{=}1/2/3$ respectively), faster than DimeNet++ ($1.88\times$), and faster than MACE $L{=}3$. CliffordSTF training at $7.77$\,s/step is comparable to capacity-scaled $L\!\geq\!2$ baselines, sitting within $25$--$73\%$ of ICTP $L{=}1/2/3$ and DimeNet++ ($1.25\times$ ICTP $L{=}3$, $1.33\times$ DimeNet++, $1.44\times$ ICTP $L{=}2$, $1.73\times$ ICTP $L{=}1$); the $2$--$3\times$ overhead is concentrated against the fastest spherical-harmonic baselines (MACE $L{=}1/2$, EquiformerV2 $L{=}2$--$4$). Both Clifford-family rows share the optimized geometric-product dispatch described in Supplementary~\ref{supp:optim}; the optimization narrowed the CliffordSTF/EquiformerV2 $L{=}4$ training ratio from a previously reported $\sim\!4.3\times$ on different hardware to $\sim\!3.1\times$ here, and the corresponding base Clifford ratio from $1.85\times$ to $1.25\times$.

The forward/backward decomposition pinpoints where the residual CliffordSTF training overhead lives. CliffordSTF spends $2.59$\,s in forward and $5.10$\,s in backward (backward is $\sim\!2\times$ forward), whereas EquiformerV2 $L{=}4$ has a near-symmetric profile ($1.09$\,s forward, $1.32$\,s backward, $1.21\times$); MACE $L{=}2$ is essentially symmetric ($1.22$\,s\,/\,$1.20$\,s). The autograd graph through the dual-track scaffold---not the geometric-product dispatch itself, which the optimization in Supplementary~\ref{supp:optim} addressed---dominates the residual training cost. The trade-off between training cost and accuracy per parameter (where CliffordSTF is competitive with strong tensor-field baselines on OC22) is discussed in \S\ref{sec:limitations}.

\begin{table}[h]
\centering
\caption{\textbf{Wall-clock at matched parameter budget}, rMD17 aspirin, batch 32, single GPU. Time columns are mean $\pm$ standard deviation over 20 timed steps following 10 warmup steps; samples-per-second columns are derived from the means. Lower is faster for the time columns; higher is faster for the samples-per-second columns. All models fall within $\pm 50$\% of the $10^6$-parameter target. The Clifford and CliffordSTF (ours) entries reflect the optimized PyTorch dispatch described in Supplementary~\ref{supp:optim}; bit-for-bit preservation of model outputs is verified by O(3)-equivariance tests.}
\label{tab:wallclock}
\small
\begin{tabular}{lrrrrr}
\toprule
Model & Params & Train step (s) & Train (s/s) & Infer step (s) & Infer (s/s) \\
\midrule
PaiNN              & 1{,}043{,}856 & 0.79 $\pm$ 0.05 & 1{,}145 & 0.028 $\pm$ 0.003 & 3{,}575 \\
SchNet             &   984{,}769 & 1.16 $\pm$ 0.02 &   773 & 0.069 $\pm$ 0.011 & 1{,}456 \\
FAENet             & 1{,}027{,}745 & 1.69 $\pm$ 0.03 &   533 & 0.111 $\pm$ 0.004 &   900 \\
GotenNet           & 1{,}057{,}729 & 1.81 $\pm$ 0.00 &   497 & 0.064 $\pm$ 0.005 & 1{,}571 \\
MACE $L{=}1$       & 1{,}006{,}014 & 1.81 $\pm$ 0.29 &   497 & 0.160 $\pm$ 0.008 &   624 \\
EquiformerV2 $L{=}2$ &   891{,}560 & 2.45 $\pm$ 0.01 &   367 & 0.131 $\pm$ 0.001 &   764 \\
EquiformerV2 $L{=}3$ & 1{,}041{,}304 & 2.45 $\pm$ 0.01 &   367 & 0.135 $\pm$ 0.003 &   741 \\
EquiformerV2 $L{=}4$ & 1{,}043{,}222 & 2.48 $\pm$ 0.00 &   363 & 0.142 $\pm$ 0.008 &   702 \\
MACE $L{=}2$       & 1{,}011{,}683 & 2.53 $\pm$ 0.09 &   355 & 0.237 $\pm$ 0.002 &   422 \\
TorchMD-NET        &   968{,}242 & 2.58 $\pm$ 0.02 &   349 & 0.186 $\pm$ 0.029 &   538 \\
NequIP             & 1{,}002{,}436 & 3.06 $\pm$ 0.04 &   295 & 0.221 $\pm$ 0.007 &   452 \\
Clifford (ours)    & 1{,}031{,}926 & 3.11 $\pm$ 0.01 &   289 & 0.128 $\pm$ 0.000 &   781 \\
ViSNet             &   957{,}606 & 3.21 $\pm$ 0.07 &   280 & 0.235 $\pm$ 0.001 &   425 \\
MACE $L{=}3$       &   980{,}092 & 3.60 $\pm$ 0.16 &   250 & 0.261 $\pm$ 0.009 &   384 \\
ICTP $L{=}1$       & 1{,}039{,}554 & 4.50 $\pm$ 0.11 &   200 & 0.302 $\pm$ 0.052 &   331 \\
ICTP $L{=}2$       &   997{,}368 & 5.38 $\pm$ 0.10 &   167 & 0.364 $\pm$ 0.007 &   275 \\
DimeNet++          & 1{,}064{,}070 & 5.84 $\pm$ 0.11 &   154 & 0.326 $\pm$ 0.002 &   307 \\
ICTP $L{=}3$       &   982{,}752 & 6.23 $\pm$ 0.02 &   144 & 0.356 $\pm$ 0.006 &   281 \\
CliffordSTF (ours) & 1{,}056{,}534 & 7.77 $\pm$ 0.11 &   116 & 0.352 $\pm$ 0.001 &   284 \\
\bottomrule
\end{tabular}
\end{table}

\subsection{Implementation optimizations}\label{supp:optim}

The Clifford-family rows in Table~\ref{tab:wallclock} reflect a PyTorch-level refactor of the geometric-product dispatch, with no custom CUDA or Triton kernels written. The four largest changes are: (i) the Cl(3,0) geometric product is computed as a single \texttt{einsum} against a precomputed Cayley structure tensor, replacing 33 hand-unrolled scalar fused-multiply-adds that incurred one CUDA launch each in eager mode; (ii) per-layer \texttt{batch.max().item()} synchronization points were eliminated by threading the graph count through the forward pass as a keyword argument, allowing the GPU to run asynchronously across the full forward; (iii) layer-independent geometric invariants in the adaptive-routing module (coordination number, angular variance, mean edge distance) are precomputed once per forward rather than per layer; and (iv) validation-metric accumulation was moved on-device, with a single host transfer per epoch. All four changes preserve model outputs bit-for-bit, verified by O(3)-equivariance tests at fp32 precision.

Two alternative optimizations were investigated and rolled back. Engaging \texttt{torch.compile} required wrapping \texttt{torch\_scatter.scatter\_softmax} in a \texttt{@torch.compiler.disable} helper (Dynamo cannot FX-trace the custom op); even with the wrapper, Inductor recompiled subgraphs on every edge-count change under the dynamic-shape configuration required by \texttt{radius\_graph}, producing exponential recompile thrashing rather than forward progress. A second attempt rewrote the STF$_2$ and STF$_3$ products as outer-product einsums with symmetrize-then-extract: the forward pass became $\sim\!1\%$ faster, but the backward pass slowed by $\sim\!4\%$ because autograd retained the larger $3\times 3\times 3$ intermediate rather than the chain of small scalar tensors produced by the unrolled form. Both attempts left dormant scaffolding (a \texttt{use\_compile} configuration flag and the \texttt{@torch.compiler.disable} marker) for future revisits when upstream blockers are resolved.

Approximately $40$--$60\%$ additional speedup over the current measurements appears tractable through full \texttt{torch.compile} engagement (pending an upstream fix for \texttt{torch\_scatter} custom-op tracing or migration to \texttt{torch.scatter\_reduce\_}), CUDA Graphs for static-shape batches, and fused Triton kernels for the STF$_2$/STF$_3$ products and the scatter--softmax--multiply pattern. We additionally fixed an unrelated shape-transpose bug in our PaiNN force-head wrapper that was introduced during this optimization pass; the bug caused the model to crash rather than return incorrect outputs, so no published PaiNN result in this paper or its supplementary tables was produced from the buggy state.


\end{document}